\documentclass[11pt, draftclsnofoot, onecolumn]{IEEEtran}
\usepackage[utf8]{inputenc}  
\usepackage{url}
\usepackage{xcolor}
\usepackage{ifthen}
\usepackage{forest}
\usepackage{cite}
\usepackage[cmex10]{amsmath} 

\usepackage{amsthm,amssymb,amsfonts}
\usepackage{mathrsfs}
\usepackage{algpseudocodex}
\usepackage{algorithm}
\usepackage{multirow} 
\usepackage{multicol}
\usepackage{makecell}
\usepackage{tabularx}
\usepackage{booktabs}

\usepackage[colorlinks, linkcolor=blue, anchorcolor=blue, citecolor=blue]{hyperref}
\usepackage{enumitem}
\usepackage{braket} 
\usepackage{blkarray}

\usepackage{pgfplots}
\pgfplotsset{compat=1.12}
\usepackage{tikz}
\usetikzlibrary{arrows.meta} 
\usetikzlibrary{positioning} 
\usepackage{tikz-cd}
\usepackage{verbatim} 
\usetikzlibrary{shapes,arrows, arrows.meta,positioning, decorations.pathreplacing,decorations.markings,decorations.pathmorphing,calc,shapes.geometric,fit,math}  

\usepackage{mathtools}

\theoremstyle{definition}
\newtheorem{theorem}{Theorem}
\newtheorem{assumption}{Assumption}

\newtheorem{lemma}{Lemma}

\newtheorem{proposition}{Proposition}

\newtheorem{example}{Example}
\newtheorem{remark}{Remark}
\newtheorem{condition}{Condition}

\newtheorem{definition}{Definition}
\newtheorem{problem}{Problem}

\newcommand{\calE}{\mathcal{E}}
\newcommand{\calF}{\mathcal{F}}

\newcommand{\calH}{\mathcal{H}}
\newcommand{\calI}{\mathcal{I}}

\newcommand{\calK}{\mathcal{K}}

\newcommand{\calN}{\mathcal{N}}

\newcommand{\calQ}{\mathcal{Q}}

\newcommand{\calX}{\mathcal{X}}

\newcommand{\calZ}{\mathcal{Z}}

\newcommand{\bfy}{\mathbf{y}}

\newcommand{\bfV}{\mathbf{V}}

\newcommand{\bfx}{\mathbf{x}}

\newcommand{\bfz}{\mathbf{z}}

\newcommand{\bfv}{\mathbf{v}}
\newcommand{\bfP}{\mathbf{P}} 
\newcommand{\bfQ}{\mathbf{Q}}

\newcommand{\bbF}{\mathbb{F}}
\newcommand{\bbN}{\mathbb{N}}

\newcommand{\Alice}[1]{\text{Alice}_{#1}}

\newcommand{\Tr}{\textrm{Tr}}

\begin{document}

\title{Precoding-based protocols for entanglement assisted linear computation over a quantum many-to-one network}   

\author{ 
  Ruoyu Meng \IEEEmembership{Graduate Student Member, IEEE}, and Aditya Ramamoorthy \IEEEmembership{Senior Member, IEEE} 

    \thanks{The material in this work has appeared in part at the 2026 IEEE International Symposium on Information Theory, Guangzhou, China. The authors are with the Department of Electrical and Computer Engineering, Iowa State University, Ames, IA, U.S.A. (Email:\{rmeng, adityar\}@iastate.edu).}
}

\maketitle 

\begin{abstract} 

In this work, we consider the problem of computing a linear combination over a noiseless quantum many-to-one network. There are $k$ senders, Alice$_1$, $\dots$, Alice$_k$, and a single receiver, Bob. Each Alice$_i$ has a data vector $W_i \in \bbF^{m_i}$ ($\bbF$ is a finite field). Bob wants to compute a linear combination (LC) $Y = \bfV_1 W_1 + \bfV_2 W_2 + \dots +\bfV_k W_k\in\bbF^m$, where $\bfV_i$ is a $m \times m_i$ matrix over $\bbF$.  The senders transmit quantum states to Bob through a noiseless many-to-one quantum network, but they are not allowed to communicate with each other. The senders share entanglement amongst themselves (Bob does not share the entanglement). They encode their classical information $W_i, i = 1, \dots, k$ into their local subsystems and transmit it to Bob such that Bob can recover $Y$, upon subsequent quantum measurement and post-processing. The N-Sum Box protocol proposed in Allaix et al. '25, considers this problem under certain constraints on the linear combination and the distribution of data vectors among the senders.

Our work presents protocols that support the computation of a more general class of linear transformations by giving the senders access to more qudits and also allowing them to judiciously precode their input symbols. The communication cost of our schemes is at most the cost of the best-known prior results in this area, and strictly lower in certain cases. Finally, we demonstrate that the communication cost is subadditive with respect to the instances. Specifically, we find two different linear functions such that the total cost of computing them individually is strictly larger than the cost of computing them jointly.

\end{abstract}

\section{Introduction} 

Distributed computation is a critical component of various technologies, e.g., scientific computing, training of deep neural networks, MapReduce \cite{DeanG08}, and sensor networks, among others. It also plays a key role within communication tasks over wireless channels (exploiting the inherent superposition of the medium) \cite{wilson_etal10, NazerG07}. 

Any form of distributed computation typically requires both computation resources and communication resources. For instance, within MapReduce \cite{DeanG08}, it is well recognized that the overall job computation time includes the computation time of the Map and the Reduce steps and the communication time of the intermediate Shuffle phase which involves network communication between the different workers. Likewise, the distributed training of neural networks also involves computations at the workers and communication either between the workers or between the workers and a designated parameter server \cite{langer2020distributed}. Broadly speaking, techniques for improving the resource-efficiency of distributed computation are of great interest. 

Ideas from coding theory have proven successful in addressing problems of these types, e.g., within MapReduce, these ideas \cite{LiMYA18, kostasR20} allow for a reduction in the induced network traffic in the Shuffle phase. Over a multiple-access channel, the work of \cite{wilson_etal10, NazerG07} shows that structured coding can exploit signal superposition to decode a desired function (e.g., a sum) directly, often outperforming decode-then-compute. Within distributed computation, the class of linear (or bilinear) functions is a large and important class that is often of interest, e.g., in gradient coding \cite{tandon_gradient}, the goal is to recover the sum of partial gradients. Distributed matrix-vector and matrix-matrix multiplication are examples of problems which can be posed as linear (or bilinear) problems \cite{RamDT20}. Likewise, computation of linear functions over finite fields using network coding has been considered in \cite{ramamoorthyL13, tripathyR18}. 



Recent progress in network quantum information theory has been pushing quantum communications beyond point-to-point links toward networked operation. Quantum teleportation and entanglement serve as central resources for connectivity and coordination \cite{CacciapuotiCVH20}, and networking challenges for distributed quantum computing have been articulated in the surveys \cite{CacciapuotiCTCGB20, MarcelloMDJAA24}. Moreover, quantum network coding has been developed and experimentally validated as a mechanism to improve network efficiency \cite{HayashiINRY07, Lu2019}.

A natural question in this domain is how quantum resources can improve the efficiency of distributed computation. The well-known superdense coding technique \cite{wilde_17} already demonstrates that for point-to-point transmission (Alice to Bob), quantum entanglement allows us to double the number of classical bits from Alice to Bob. In the same point-to-point setting, our prior work \cite{mengR25} shows that quantum protocols can be exponentially better than classical protocols when Bob has side information correlated with Alice and wants to compute a function of Alice's message and his own side information.
The work of \cite{LiuLTL02, ShadmanKBM12} shows that using shared entanglement, distributed superdense coding can reduce the communication cost within certain information processing tasks. Shared quantum entanglement was shown to improve the performance of private information retrieval (PIR) problems \cite{SongH_collude_21}. More recently, the work of \cite{N_Sum_Box_JafarYao25} considers a problem of distributed linear computation over quantum many-to-one networks, and demonstrates the advantages of protocols that exploit quantum entanglement between the distributed senders.


\subsection{Motivation and Background}

The motivation of our work originates from the $N$-Sum Box protocol \cite{N_Sum_Box_JafarYao25}, which is designed for solving a distributed superdense coding problem. This problem is formulated as follows. 
Let $\bbF_q$ be the finite field of order $q$. Suppose there are $N$ senders (Alice$_1$, $\dots$, Alice$_N$) and a receiver (Bob). There is a noiseless $N$-to-$1$ quantum  channel from the senders to the receiver. For $i\in[N]$ ($[N]$ denotes the set $\{1, \dots, N\}$), Alice$_i$ has two inputs $x_{i},z_{i}\in\bbF_q$. Denote $\bfx=[x_1,\dots,x_N]^T,\bfz=[z_1,\dots,z_N]^T$. Bob wants to compute $\bfy = M_x \bfx + M_z \bfz$ where $M_x,M_z$ are $\kappa \times N$ matrices over $\bbF_q$. The senders share a joint quantum system $\calQ=Q_1\dots Q_N$ where Alice$_i$ has the $q$-dimensional subsystem $Q_i$. For each $i\in [N]$, Alice$_i$ encodes $[x_i,z_i]$ using her subsystem $Q_i$ and then transmits $Q_i$ to Bob through the noiseless quantum channel. 

The original $N$-Sum Box protocol of~\cite{N_Sum_Box_JafarYao25} gives a stabilizer-based construction (under certain conditions on $[M_x ~|~ M_z]$) when $\kappa = N$
%
whereby Bob obtains a length-$N$ vector $\bfy$ over $\bbF_q$ from $N$ transmitted qudits. The same work also gives two extensions that are important for us.  First, the $N$-Sum Box allows for a locally invertible transform (LIT) freedom: one may locally rescale one side of the input pair by an invertible diagonal matrix before applying the Sum Box construction.  Second, a $\kappa$-output computation with $\kappa\le N$ can be reduced to an $N$-output Sum Box by appending appropriately chosen $N-\kappa$ auxiliary output coordinates. Bob discards the auxiliary values after computing the $N$ outputs.

In the operational form used in this paper, the $(\kappa,N)$-Sum Box primitive computes
\[
    \bfy=M_x\bfx+M_z\bfz,\qquad
    M_x,M_z\in\bbF_q^{\kappa\times N},
\]
whenever matrices $M_x,M_z$ satisfy the rank condition
\[
    \operatorname{rank}([M_x|M_z])=\kappa,\qquad \kappa\le N,
\]
and the strong self-orthogonality condition
\[
    \Omega(M_x,M_z):=M_xM_z^T-M_zM_x^T=0.
\]
We call this a $(\kappa,N)$-Sum Box protocol.  If $\kappa=N$, then this reduces to the $N$-output Sum Box setting. 

The SSO condition restricts the class of linear transformations that Bob can compute by the $N$-Sum Box protocol. Thus, a natural question of interest is how one can compute the linear transformations when the SSO condition does not hold. This issue was investigated in part in the work of \cite{HuUlukus25}, which also considered a scenario where multiple instances of the linear transformation need to be computed in a block. In this work, we investigate several facets of this problem and present our findings. 
\subsection{Illustrative Examples}
\begin{example}\label{example:eg1}
Consider two senders Alice$_i$ with inputs $x_i,z_i \in \bbF_2$ for $i=1,2$ and suppose that Bob wishes to recover
\begin{align*}
    \begin{bmatrix}
        y_1\\y_2
    \end{bmatrix} =  
        \underbrace{\begin{bmatrix}
            1&0 \\
            0&1
        \end{bmatrix}}_{M_x}\begin{bmatrix}
        x_1\\x_2 
    \end{bmatrix} +\underbrace{\begin{bmatrix}
            0&1\\
            1&0
        \end{bmatrix}}_{M_z}\begin{bmatrix}
        z_1\\z_2
    \end{bmatrix}.
\end{align*}
As $M_x = I_2$ and $M_z = M_z^T$ we can see that $M_xM_z^T-M_zM_x^T = 0$ and it is easy to see that $\text{rank}([M_x|M_z]) = 2$.
It follows that this computation can be performed by using a 2-Sum Box protocol.
\end{example}

\begin{example}
Now suppose there are four senders such that Alice$_i$ (for $i = 1, \dots, 4$) possesses $(x_i, z_i) \in \bbF_2$. Bob wants to compute 
\begin{align}
    \begin{bmatrix}
        y_1\\y_2\\y_3\\y_4
    \end{bmatrix} =  
        \underbrace{\begin{bmatrix}
            1&0&0&0 \\
            0&1&0&0\\
            0&0&1&0\\
            0&0&0&1
        \end{bmatrix}}_{M_x}\begin{bmatrix}
        x_1\\x_2\\x_3\\x_4 
    \end{bmatrix} +\underbrace{\begin{bmatrix}
            0&1&1&1\\
            0&1&0&1\\
            1&1&0&1\\
            1&1&0&1
        \end{bmatrix}}_{M_z}\begin{bmatrix}
        z_1\\z_2\\z_3\\z_4 
    \end{bmatrix}. \label{eq:eg_2}
\end{align}
Here $M_x = I_4$ and $M_z \neq M_z^T$, so that $M_xM_z^T - M_zM_x^T \neq 0$, i.e., the SSO condition is not satisfied and hence the 4-Sum Box protocol does not apply.
However, we now show that if one sender is given additional qubits, then we can in fact compute the required linear transformation. Towards this end,  we note that \eqref{eq:eg_2} can be written
\begin{align*} 
    \begin{bmatrix}
        y_1\\y_2\\y_3\\y_4
    \end{bmatrix}
    &=\underbrace{\begin{bmatrix}
            1&0&0&0&1 \\
            0&1&0&0&0\\
            0&0&1&0&0\\
            0&0&0&1&0
        \end{bmatrix}}_{M_x'}\underbrace{\begin{bmatrix}
        x_1\\x_2\\x_3\\x_4 \\\textcolor{red}{z_3}
    \end{bmatrix}}_{\bfx'}+\underbrace{\begin{bmatrix}
            0&1&1&1&0\\
            0&1&1&1&1\\
            1&1&0&1&0\\
            1&1&1&1&0
        \end{bmatrix}}_{M_z'}\underbrace{\begin{bmatrix}
        z_1\\z_2\\ \textcolor{red}{0} \\z_4 \\\textcolor{red}{0}
    \end{bmatrix}}_{\bfz'}.
\end{align*}
It can be verified that $\text{rank}([M_x'| M_z']) = 4 \le 5$ (setting $\kappa=4,N=5$) and
\begin{align*}
    &M_x'(M_z')^T - M_z'(M_x')^T\\
    =&\begin{bmatrix}
            1&0&0&0&1 \\
            0&1&0&0&0\\
            0&0&1&0&0\\
            0&0&0&1&0
        \end{bmatrix}\begin{bmatrix}
            0&0&1&1\\
            1&1&1&1\\
            1&1&0&1\\
            1&1&1&1\\
            0&1&0&0
        \end{bmatrix} -\begin{bmatrix}
            0&1&1&1&0\\
            0&1&1&1&1\\
            1&1&0&1&0\\
            1&1&1&1&0 
        \end{bmatrix}\begin{bmatrix}
            1&0&0&0 \\
            0&1&0&0\\
            0&0&1&0\\
            0&0&0&1 \\
            1&0&0&0
        \end{bmatrix}\\
        =& \begin{bmatrix}
            0&1&1&1\\
            1&1&1&1\\
            1&1&0&1\\
            1&1&1&1\\
        \end{bmatrix} -  \begin{bmatrix}
            0&1&1&1\\
            1&1&1&1\\
            1&1&0&1\\
            1&1&1&1\\
        \end{bmatrix}=0.
\end{align*}
In this way a $(4,5)$-Sum Box protocol can be employed for $(M_x',M_z')$ which in turn can be used to compute $y=M_x\bfx+M_z\bfz$. Consider five entangled qubits $Q_i, i = 1,\dots, 5$, where Alice$_1$, Alice$_2$ and Alice$_4$ still have $Q_1,Q_2$ and $Q_4$ respectively, but Alice$_3$ has $Q_3$ and $Q_5$. Alice$_i$ encodes $[x_i,z_i]$ to $Q_i$ for $i\in\{1,2,4\}$. For Alice$_3$, she encodes $[x_3,0]$ to $Q_3$ and $[z_3,0]$ to $Q_5$. Then, they run the $(4,5)$-Sum Box protocol and Bob can obtain the desired result.

\end{example}
\begin{example}\label{example:motivating_example3}
This example is inspired by the work of \cite{Summation_JafarYao25}. We now consider two senders, Alice$_i$, each with a single bit $x_i \in \bbF_2$, $i=1,2$ and suppose that Bob wants to compute
\begin{align*}
    y = x_1+x_2.
\end{align*}
For a single instance, the trivial one-shot protocol has Alice$_1$ send $x_1$ and Alice$_2$ send $x_2$, using two qudits in total. A 2-Sum Box protocol, with $\kappa=1,N=2$, is also possible, but it has the same one-shot communication cost. The point of the next construction is that coding across two instances reduces the normalized cost to one qudit per instance.
\begin{align*}
    \begin{bmatrix}
        y^{(1)}\\y^{(2)}
    \end{bmatrix} = \begin{bmatrix}
        x_1^{(1)}+x_2^{(1)}\\x_1^{(2)}+x_2^{(2)}
    \end{bmatrix} = \begin{bmatrix}
        1 & 0 \\ 0 & 1
    \end{bmatrix} \begin{bmatrix}
        x_1^{(1)}\\x_2^{(2)}
    \end{bmatrix} +  \begin{bmatrix}
        0 & 1 \\ 1 & 0
    \end{bmatrix} \begin{bmatrix}
      x_1^{(2)} \\  x_2^{(1)}  
    \end{bmatrix},
\end{align*}
where the superscript denotes the instance index, and Alice$_i$ has $\begin{bmatrix}
     x_i^{(1)}\\ x_i^{(2)}
\end{bmatrix}$ for $i=1,2$. 
Then, we are back in the setting of Example \ref{example:eg1}, and can thus use a 2-Sum Box for the desired computation. Therefore, the total number of qubits transmitted is two, i.e. one qubit per instance. This example shows that considering multiple instances may reduce the communication cost per instance. 
\end{example}

\subsection{Related Work}

Our model is related to, but distinct from, the standard quantum multiple-access channel literature. In a standard quantum multiple-access channel, several transmitters communicate through a common quantum channel, and the main objective is typically to characterize achievable rate regions for message transmission. In contrast, we consider a noiseless quantum many-to-one network with entanglement shared among the senders, and the objective is exact linear function computation at the receiver with minimum qudit communication. Capacity regions for classical-quantum and quantum multiple-access channels were studied in \cite{Winter01,YardHD08}. Entanglement-assisted variants were considered in \cite{ShiHGZZ21,HsiehDW08}, where single-letter characterizations were obtained in certain cases. 

A related line of work on entanglement-assisted many-to-one quantum communication uses the stabilizer formalism~\cite{CalderbankRSS97, AshikhminK01, KetkarKKS06, Gottesman97}, which originally was used for  quantum error correction and fault-tolerant quantum computing~\cite{Gottesman97, Calderbank96, BennettDSJW96}. Private Information Retrieval (PIR)\cite{ChorKGS98} enables a user to retrieve a desired file from distributed storage without revealing its identity to any individual server. In Quantum PIR (QPIR), the answers are quantum systems and the servers may share prior entanglement. The work of \cite{SongH_multiple_21} showed that, with entanglement across databases, the QPIR capacity for replicated storage collapses to 1, independent of the number of servers and files, and a rate-one protocol is achievable already with two servers, with a strong converse bound. Other variants of QPIR were considered in \cite{SongH_QPIR_SYMMETRIC21, SongH_collude_21,  AllaixSHPHH22}. 

Entanglement-assisted many-to-one quantum communication naturally enables linear computation at the receiver. This capability appears implicitly in the QPIR literature \cite{SongH_multiple_21, SongH_collude_21} via the 2-Sum Box. The study of linear computation over a noiseless quantum many-to-one network with entangled transmitters has recently led to several capacity characterizations and coding constructions\cite{N_Sum_Box_JafarYao25, HuUlukus25, Inverted_3_Sum_Box_JafarYao25}, which are the works closest to ours. The  $N$-Sum Box protocol \cite{N_Sum_Box_JafarYao25} generalizes the 2-Sum Box instance in \cite{SongH_multiple_21, SongH_collude_21}; it provides a black-box linear-computation primitive over many-to-one quantum networks. Specifically, the $N$-Sum Box  can be viewed as a coding scheme for the set of functions restricted by the SSO condition.  The summation problem over entanglement-assisted many-to-one quantum networks, referred to as $\Sigma$-QMAC in \cite{Summation_JafarYao25}, has been characterized for arbitrary replication and entanglement patterns; this is a specific class of functions in our setting. 

In this work, we consider general linear combination problems.  The vector linear computation problem has been fully solved for  $k=3$ senders\cite{Inverted_3_Sum_Box_JafarYao25}. However, the computational complexity of their method appears to grow exponentially fast in $k$. 
In \cite{HuUlukus25}, an entanglement-assisted stabilizer-based scheme has been proposed for arbitrary linear computations, achieving capacity in certain cases. However, their coding scheme contains an optimal precoding problem, which closely resembles MinRank (NP-hard \cite{FaugereJEMSP10}), but no algorithm is known for the optimal precoding problem. 


\subsection{Contribution and Organization}
We study the entanglement-assisted linear combination  problem over a $k$-to-$1$ noiseless quantum network, where the senders may share arbitrary entanglement and the communication cost is the total number of transmitted $q$-dimensional quantum systems (qudits). Each Alice$_i$ has a data vector $W_i\in\bbF_q^{m_i}$ and Bob wants to compute $\bfy = \sum_{i\in[k]} \bfV_iW_i\in\bbF_q^m$ for some matrices $\bfV_i\in\bbF_q^{m\times m_i}$. Our first contribution is the design of three new achievable schemes, each operating under a distinct set of structural conditions on $(\bfV_1,\dots,\bfV_k)$, thereby relaxing restrictions required by \cite{N_Sum_Box_JafarYao25}. Second, we provide provable comparisons with the best-known prior result of \cite{HuUlukus25}: one of our schemes has no larger total cost than the scheme of \cite{HuUlukus25} (and is strictly better for some instances), another yields strict improvements on a specific subclass, and a third scheme is shown to be information-theoretically optimal in its stated regime via a matching converse. Third, we demonstrate a genuinely quantum subadditivity (``bundling'') phenomenon, showing that joint coding across multiple linear combination problems  can strictly reduce the total qudit cost, while proving that the classical counterpart is additive. 

This paper is organized as follows. The problem formulation appears in Section \ref{sec:prob_formulation}. In Section \ref{sec:statement_main_contributions}, we formally state our main results and discuss their consequences. Section \ref{sec:preliminaries} overviews preliminary ideas and Section \ref{sec:main_contribution_formal_statement_and_proof} discusses the proofs of main results. Section \ref{sec:comparison_with_prior_works} provides quantitative comparisons with prior work, and Section \ref{sec:conclusions} concludes the paper with a discussion of future work.


\section{Problem Formulation}\label{sec:prob_formulation}

A linear combination (LC) problem can be represented by a tuple $(\bbF_q, k, \bfV_1,\dots,\bfV_k)$, where there are $k$ senders  (denoted as Alice$_i$, $i\in [k]$). For each $i\in   [k]$, $\bfV_i$ is an $m\times m_i$ matrix with elements in $\bbF_q$. Alice$_i$ has data vector $W_i$, which is an arbitrary vector of length $m_i$ over $\bbF_q$. The receiver, Bob, wants to recover the following linear function of these data vectors.
\begin{align*}
    Y = \bfV_1 W_1+\dots + \bfV_k W_k.
\end{align*}
The computation is assumed to occur over multiple instances. Suppose that $L$ instances of this function need to be computed by Bob, i.e., he wants to recover 
\begin{equation}
    \begin{split}
        Y^{(1)} =& \bfV_1 W_1^{(1)}+\dots + \bfV_k W_k^{(1)},\\
        \vdots&\\
        Y^{(L)} =& \bfV_1 W_1^{(L)}+\dots + \bfV_k W_k^{(L)}.
    \end{split}
\end{equation}
We represent this compactly as
\begin{equation}
    Y^{[L]} = \bfV_1^{[L]} W_1^{[L]}+\dots + \bfV_k^{[L]} W_k^{[L]}
\end{equation}
where
\begin{align*}
   Y^{[L]} = \begin{bmatrix}
        Y^{(1)}\\
        \vdots\\
        Y^{(L)}
    \end{bmatrix} \text{ and }\forall i\in[k],\,\bfV^{[L]}_i =\underbrace{\begin{bmatrix}
        \bfV_i&&\\
        &\ddots &\\
        &&\bfV_i
    \end{bmatrix} }_{L\text{ times}},\,W_i^{[L]}=
    \begin{bmatrix}
        W_i^{(1)}\\
        \vdots\\
        W_i^{(L)}
    \end{bmatrix}.
\end{align*} 
In the above expression (and throughout the manuscript), unspecified entries in matrices are assumed to be zero. We now state an assumption on the $\bfV_i$ matrices and show that it holds without loss of generality.
\begin{assumption}\label{assumption:lin_ind_assumption}
\begin{enumerate}[label = (\arabic*)]
    \item For each $i\in[k]$, the columns  of $\bfV_i$ are all linearly independent, i.e., $\bfV_i$ has full column rank.
    \item $\text{rank}([\bfV_1|\bfV_2|\dots|\bfV_k]) = m$, i.e., $[\bfV_1|\bfV_2|\dots|\bfV_k]$ has full row rank.   
\end{enumerate}
\end{assumption}

\begin{remark}[Justification of Assumption~\ref{assumption:lin_ind_assumption}]\label{remark:assumption_proof} 
We may always reduce an LC instance to satisfy Assumption~\ref{assumption:lin_ind_assumption} without changing the optimal communication cost.

\smallskip
\noindent\textbf{(1) Full column rank of each $\bfV_i$.}
If $\bfV_i$ is not full column rank, let $r_i=\text{rank}(\bfV_i)$.
Choose a submatrix $\bfV_i'\in\bbF_q^{m\times r_i}$ whose columns form a basis of the columns of $\bfV_i$.
Then there exists a matrix $P_i\in\bbF_q^{r_i\times m_i}$ such that $\bfV_i=\bfV_i'P_i$.
Hence,
\[
\bfV_i W_i=\bfV_i'(P_i W_i).
\]
Thus, Alice$_i$ can locally precode $W_i' := P_i W_i$ and we obtain an equivalent LC instance with matrix $\bfV_i'$ and data vector $W_i'$ at Alice$_i$. Conversely, since $P_i$ has full row rank, Alice$_i$ can lift any reduced input W$_i'$ to a preimage under $P_i$ and then run the original coding scheme.

\smallskip
\noindent\textbf{(2) Full row rank of $[\bfV_1|\cdots|\bfV_k]$.}
Suppose $\text{rank}([\bfV_1|\cdots|\bfV_k]) = r < m$. Then, there exist matrices
$D\in\mathbb{F}_q^{m\times r}$ with full column rank  and $\bfV_i'\in\mathbb{F}_q^{r\times m_i}$
such that $[\bfV_1|\cdots|\bfV_k] = D[\bfV_1'|\cdots|\bfV_k']$.
Let $Y' := \sum_{i\in[k]}\bfV_i' W_i \in \mathbb{F}_q^{r}$.
Any coding scheme that enables Bob to recover $Y'$ also enables him to recover
$Y = \sum_{i\in[k]}\bfV_i W_i$ via post-processing by Bob to obtain $Y = DY'$.
Hence,  we may assume $\mathrm{rank}([\bfV_1|\cdots|\bfV_k]) = m$. Conversely, since $D$ has full column rank, Bob can recover $Y '$ from $Y=DY '$ by applying a left inverse of $D$.

Hence, both reductions preserve the optimal communication cost.
\end{remark}


\subsection{Communication Model and Coding Schemes}\label{subsec:coding_scheme}

 

We next describe the communication model for an LC problem. The senders are not allowed to communicate with each other. Communication takes place over a noiseless many-to-one network: each Alice$_i$ has a one-way link to Bob and
transmits a system only to Bob.

We consider two variants of this model. In the quantum variant, the senders may share an arbitrary joint quantum state before their inputs are realized. 
Bob does not share entanglement with the senders. After observing her data vector, each Alice$_i$ applies a local input-dependent unitary to her subsystem and transmits the resulting quantum system to Bob through her noiseless quantum link. We refer to this as the \emph{noiseless quantum many-to-one network model}, or simply the \emph{quantum many-to-one model} when the noiseless nature of the links is clear.

In the classical variant, there is no entanglement among the senders, and each Alice$_i$ transmits a string of $q$-ary symbols to Bob through a noiseless classical link. We refer to this as the \emph{classical many-to-one model}.  

Given an LC $(\bbF_q, k, \bfV_1,\dots,\bfV_k)$ with the noiseless quantum many-to-one model, the coding scheme is specified as follows. 
\begin{itemize}
    \item The number of instances supported by the scheme, denoted by an integer
    $L\ge 1$.

    \item A joint quantum system $\calQ=\calQ_1\cdots\calQ_k$ with initial state
    $\rho_{\mathrm{init}}$. Each subsystem $\calQ_i$ is held by Alice$_i$ and may
    consist of multiple $q$-dimensional qudits. We denote by $\delta_i$ the
    number of qudits in $\calQ_i$.

    \item A collection of local unitary encoding maps
    \[
        \left\{\calE_i^{(w_i^{[L]})}: i\in[k],\,
        w_i^{[L]}\in\bbF_q^{m_iL}\right\},
    \]
    where
    \[
        \calE_i^{(w_i^{[L]})}(\rho)
        =
        U_i^{(w_i^{[L]})}\rho
        \left(U_i^{(w_i^{[L]})}\right)^\dagger
    \]
    for some unitary $U_i^{(w_i^{[L]})}$ acting on $\calQ_i$. Note that $U_i^{(w_i^{[L]})}$ depends on the data vector $w_i^{[L]}$.

    \item A POVM \cite{wilde_17} $\{\Lambda_u:u\in\calI\}$ for some finite index set $\calI$,
    representing Bob's quantum measurement after receiving all transmitted
    systems.

    \item A deterministic post-processing map
    $g:\calI\to\bbF_q^{mL}$ used by Bob to produce the final output.
\end{itemize} 

The coding scheme consists of the following three stages.

\begin{enumerate}
    \item \textbf{Entanglement distribution:}
    The joint quantum system $\calQ=\calQ_1\cdots\calQ_k$ with initial state
    $\rho_{\mathrm{init}}$ is distributed so that Alice$_i$ holds subsystem
    $\calQ_i$ for each $i\in[k]$.

    \item \textbf{Encoding and transmission:}
    Suppose the realization of the $L$-instance data vectors is
    \[
        (W_1^{[L]},\dots,W_k^{[L]})
        =
        (w_1^{[L]},\dots,w_k^{[L]}).
    \]
    For each $i\in[k]$, Alice$_i$ applies the local unitary encoding map
    $\calE_i^{(w_i^{[L]})}$ to her subsystem $\calQ_i$. The joint state after
    encoding is
    \begin{equation}
        \rho^{(\mathbf{w})}
        =
        \calE_1^{(w_1^{[L]})}\otimes\cdots\otimes
        \calE_k^{(w_k^{[L]})}
        \left(\rho_{\mathrm{init}}\right).
        \label{eq:encoded_state}
    \end{equation}
    Alice$_i$ then sends her subsystem to Bob through her noiseless quantum link.
    The communication cost of Alice$_i$ is $\delta_i$ qudits.

    \item \textbf{Decoding and post-processing:}
    After receiving all transmitted systems, Bob performs the POVM
    $\{\Lambda_u:u\in\calI\}$ on $\rho^{(\mathbf{w})}$ and obtains an outcome
    $u$ with probability
    $\Tr(\Lambda_u\rho^{(\mathbf{w})})$. He then applies the deterministic
    post-processing map $g:\calI\to\bbF_q^{mL}$ and outputs
    $\widehat{Y}^{[L]}:=g(u)$.
\end{enumerate}
We require zero-error computation, i.e., for every
$\mathbf{w}=(w_1^{[L]},\dots,w_k^{[L]})$,
\begin{equation}
    \Pr\left[
        \widehat{Y}^{[L]}
        =
        \sum_{i\in[k]}\bfV_i^{[L]}w_i^{[L]}
        \,\middle|\,
        (W_1^{[L]},\dots,W_k^{[L]})=\mathbf{w}
    \right]
    =1.
\end{equation}

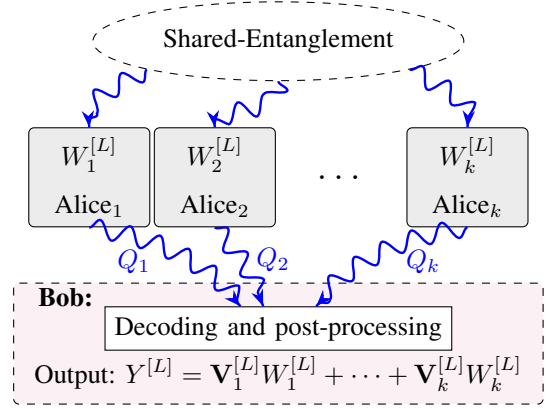
\begin{figure}[t!]
    \centering
    \begin{tikzpicture}[
    font=\small,
    >=Latex,
    dots/.style = {draw=none, fill=none, inner sep = 0pt, outer sep = 0pt, font=\Large},
    server/.style={draw, rounded corners=2pt, fill=gray!15,
                   minimum width=1.6cm, minimum height=1.1cm, align=center},
    ent/.style={draw, ellipse, dashed, inner sep=6pt, align=center},
    qsnake/.style={blue, line width=0.9pt, decorate,
                   decoration={snake, amplitude=1.2mm, segment length=4mm}},
    qsnakeStealth/.style={
    qsnake,
    postaction={
      decorate,
      decoration={
        markings,
        mark=at position 1 with {\arrow{Stealth[length=2.2mm]}}
      }}},
    qarrow/.style={blue, line width=0.9pt, -{Latex[length=2.2mm]}},
    dashedbox/.style={draw, dashed, rounded corners=2pt, inner sep=7pt, fill=purple!6}
]

\node[ent] (E) at (0,3) {Shared-Entanglement};

\node[server] (S1) at (-2.5,1.2) {$W_1^{[L]}$\\Alice$_1$};
\node[server] (S2) at (-0.84,1.2) {$W_2^{[L]}$\\Alice$_2$};
\node[dots] (S3) at (0.84,1.2) {$\dots$};
\node[server] (S4) at (2.5,1.2) {$W_k^{[L]}$\\Alice$_k$};

\draw[qsnakeStealth] (E.south west) to[out=210,in=90] (S1.north);
\draw[qsnakeStealth] (E.south)      to[out=270,in=90] (S2.north);
\draw[qsnakeStealth] (E.south east) to[out=330,in=90] (S4.north);

\node[dashedbox, minimum width=7cm, minimum height=1.6cm, align=left] (Ubox) at (0,-1) {};
\node[anchor=west] at ($(Ubox.north west)+(0.2,-0.2)$) {\textbf{Bob:}};

\node[draw, minimum width=1cm, minimum height=0.5cm, fill=white] (meas) at (0,-0.8) {Decoding and post-processing}; 

\draw[qsnakeStealth] (S1.south) to[bend left=18]
  node[pos=0.5, left, xshift=-4pt, yshift=-3pt] {$Q_1$} ($(meas.north)+(-0.5,0)$);

\draw[qsnakeStealth] (S2.south) to[bend left=10]
  node[pos=0.5, right, xshift=1pt, yshift=2pt] {$Q_2$} ($(meas.north)+(-0.2,0)$);

\draw[qsnakeStealth] (S4.south) to[bend right=18]
  node[pos=0.5, right, xshift=4pt, yshift=-3pt] {$Q_k$} ($(meas.north)+(0.5,0)$);

\node[align=left] at (0,-1.4) {$
\begin{aligned}
\text{Output: }Y^{[L]} = \bfV_1^{[L]} W_1^{[L]}+\dots + \bfV_k^{[L]} W_k^{[L]}
\end{aligned}
$};


\end{tikzpicture}
    \caption{A schematic diagram for a coding scheme under the noiseless quantum many-to-one model. Each Alice$_i$ encodes $W_i^{[L]}$ into their local system $\calQ_i$, which is transmitted to Bob. Bob applies decoding and post-processing to recover $Y^{[L]} = \bfV_1^{[L]} W_1^{[L]}+\dots + \bfV_k^{[L]} W_k^{[L]}$.}
    \label{fig:placeholder}
\end{figure}


For comparison, we also define the classical many-to-one model. In this model, there is no entanglement among the senders. For each $i\in[k]$, Alice$_i$ uses an encoding function $f_i:\bbF_q^{m_iL}\to\bbF_q^{\delta_i}$  and sends the resulting length-$\delta_i$ string of $q$-ary symbols to Bob through her noiseless classical link. Bob applies a decoding map $D:\bbF_q^{\delta_1}\times\cdots\times\bbF_q^{\delta_k}     \to \bbF_q^{mL}$  to produce $\widehat{Y}^{[L]}$. We again require zero-error computation, i.e., for every input realization, $\widehat{Y}^{[L]}
    =
    \sum_{i\in[k]}\bfV_i^{[L]}w_i^{[L]}.$


\subsection{Cost Tuple and (Optimal) Total Cost} 
Given a feasible coding scheme under a fixed model, we define the normalized cost tuple $\Delta$ and the total cost $\Gamma=\Gamma(\Delta)$ as
\begin{align*}
\Delta := \left(\frac{\delta_1}{L},\dots,\frac{\delta_k}{L}\right), \text{~and~} \Gamma(\Delta) := \sum_{i\in[k]}\frac{\delta_i}{L}.
\end{align*}
For a model $\mathsf{M}\in\{\mathsf{Q},\mathsf{C}\}$, where $\mathsf{Q}$ denotes
the quantum many-to-one model and $\mathsf{C}$ denotes the classical many-to-one
model, let $\mathbb{S}^{\mathsf{M}}_L$ denote the set of all feasible
zero-error $L$-instance coding schemes under model $\mathsf{M}$. The optimal
$L$-instance total cost under model $\mathsf{M}$ is
\[
    \Gamma^{*,\mathsf{M}}_L
    :=
    \inf_{\mathcal{S}\in\mathbb{S}^{\mathsf{M}}_L}
    \sum_{i\in[k]}\frac{\delta_i(\mathcal{S})}{L},
\]
where $\delta_i(\mathcal{S})$ is the number of qudits transmitted by Alice$_i$ in the quantum model or the number of  \(q\)-ary symbols transmitted by Alice$_i$ in the classical model. 
Finally, the optimal
asymptotic total cost is
\[
    \Gamma^{*,\mathsf{M}}
    :=
    \inf_{L\in\mathbb{N}}\Gamma^{*,\mathsf{M}}_L.
\]
When the model is clear from context, we omit the superscript $\mathsf{M}$.

\subsection{Baseline Bounds}
\label{subsec:baseline_bounds}
We record two basic bounds that will be used as benchmarks throughout the paper.
The first is a quantum converse bound from \cite{Inverted_3_Sum_Box_JafarYao25},
restated here in the terminology of the noiseless quantum many-to-one model. The
second gives the exact optimal cost under the classical many-to-one model. 
\begin{lemma} (Theorem 1, \cite{Inverted_3_Sum_Box_JafarYao25}) \label{lemma:quantum_many_to_one_lower_bound}
Let $\mathrm{LC}=(\bbF_q,k,\bfV_1,\dots,\bfV_k)$ be an LC problem satisfying Assumption \ref{assumption:lin_ind_assumption}. Under the
noiseless quantum many-to-one model, suppose there is an $L$-instance zero-error
coding scheme with normalized cost tuple
$\Delta=(\Delta_1,\dots,\Delta_k)$. Then
    \begin{align}
        \Gamma(\Delta) &\ge \text{rank}([\bfV_1|\dots|\bfV_k])=m, \text{and}\label{eq:quantum_many_to_one_lower_bound} \\
        2\sum_{i\in\calK}\Delta_i &\ge \text{rank}\left([\bfV_i]_{i\in\calK}\right) \text{ for every subset $\calK\subseteq[k]$}.
    \end{align}
\end{lemma}
\begin{proposition}[Classical many-to-one baseline]
\label{prop:classical_many_to_one_baseline}
Let $\mathrm{LC}=(\bbF_q,k,\bfV_1,\dots,\bfV_k)$ be an LC problem satisfying
Assumption~\ref{assumption:lin_ind_assumption}. Under the classical many-to-one model, we have
\begin{equation}\label{eq:prop:classical_many_to_one_baseline}
     \Gamma^{*,\mathsf{C}} = \sum_{i\in[k]} \text{rank}(\bfV_i) =  \sum_{i\in[k]}m_i.
\end{equation}
\end{proposition}

\begin{proof}
Fix $i_0\in[k]$ and consider any zero-error $L$-instance classical coding scheme. Suppose that Bob is given all data vectors $\{W_i^{[L]}:i\ne i_0\}$ as side information. This can only make the decoding problem easier. Under this side information, Bob must still recover $\bfV_{i_0}^{[L]}W_{i_0}^{[L]}$  with zero error from Alice$_{i_0}$'s message.

By Assumption~\ref{assumption:lin_ind_assumption}, $\bfV_{i_0}$ has full column rank, and therefore $\bfV_{i_0}^{[L]}W_{i_0}^{[L]}$ can take $q^{m_{i_0}L}$ distinct values as $W_{i_0}^{[L]}$ ranges over $\bbF_q^{m_{i_0}L}$. Since Alice$_{i_0}$ sends a message in $\bbF_q^{\delta_{i_0}}$, zero-error recovery requires $q^{\delta_{i_0}} \ge q^{m_{i_0}L}.$ Thus $\delta_{i_0}/L\ge m_{i_0}$. Since $i_0$ was arbitrary,
\[
    \Gamma(\Delta)
    =
    \sum_{i\in[k]}\frac{\delta_i}{L}
    \ge
    \sum_{i\in[k]}m_i.
\]

For achievability, each Alice$_i$ sends $W_i^{[L]}$ to Bob using $m_iL$
$q$-ary symbols. Bob then computes
\[
    Y^{[L]}=\sum_{i\in[k]}\bfV_i^{[L]}W_i^{[L]}.
\]
This achieves normalized total cost $\sum_{i\in[k]}m_i$ and concludes the proof.
\end{proof}
\section{Statement of Results}
\label{sec:statement_main_contributions}
In this section we provide a formal statement and discussion of our main results; the proofs appear in later sections.

\subsection{Precoding-Based Linear-Combination Protocols}
Our first set of results gives achievable schemes for LC problems in which the $N$-Sum Box protocol cannot be applied directly. A typical obstruction is the failure of the strong self-orthogonality (SSO) condition defined formally below.

\begin{definition}[SSO condition matrix]
\label{def:SSO_matrix}
For matrices $M_x,M_z\in\bbF_q^{\kappa\times N}$, define
\[
    \Omega(M_x,M_z):=M_xM_z^T-M_zM_x^T.
\]
We say that the strong self-orthogonality (SSO) condition holds if
$\Omega(M_x,M_z)=0$.
\end{definition}

To overcome this obstruction, we perform encoding across multiple instances, and utilize distributed precoding at the senders. The goal of the precoding step is to transform the original LC problem into a form for which the $N$-Sum Box protocol can be applied. We begin by defining classes of LC problems. 
\begin{definition}[\text{Interval-$[0,1]$} LC]
    An LC problem $ (\bbF_q, k, \bfV_1,\dots,\bfV_k)$ is said to be Interval-$[0,1]$ if $\frac{2m}{\sum_{i\in[k]} m_i} \in[0,1]$.
\end{definition}
We also define a subclass of LC problems, which we call Restricted Interval-$[0,1]$ LC. As this definition is slightly more involved, we defer it to Definition~\ref{def:Restricted Interval-$[0,1]$} in Section \ref{sec:main_contribution_formal_statement_and_proof}.  We next define the precoding parameter used in the first achievability result.

\begin{problem}[Optimal Precoding Problem]\label{problem:precoding} Let $\text{LC}=(\bbF_q,k,\bfV_1,\dots,\bfV_k)$ be an Interval-$[0,1]$ problem. For each integer $s\ge 0$, let $\calF(s)$ be the set of tuples $(\{P_i\}_{i=2}^k,\tilde{M}_1, \tilde{M}_2,\tilde{X}_1)$ satisfying the following conditions: 
\begin{itemize}
    \item  $P_i, i = 2, \dots, k$ range over invertible matrices in $\bbF_q^{m_i\times m_i}$.
    \item $\tilde{M}_1,\tilde{M}_2$ range over matrices in $\bbF_q^{2m\times (m_1+s)}$ and  $\tilde{X}_1$ ranges over matrices in $\bbF_q^{ 2(m_1+s)\times 2m_1}$.
    \item 
The following constraints need to be satisfied.
\begin{enumerate}[label=(\alph*)]
    \item  \[\begin{bmatrix}
        \bfV_1&\\ & \bfV_1
    \end{bmatrix} = [\tilde{M}_1|\tilde{M}_2]\tilde{X}_1.  \]
    \item \[\Omega(\tilde{M}_1,\tilde{M}_2) + \sum_{i\in[k] \setminus \{1\}}  \Omega\left(\begin{bmatrix}
            \bfV_iP_i\\ 0
        \end{bmatrix}, \begin{bmatrix}
            0 \\ \bfV_i
        \end{bmatrix} \right) =0.\]
\end{enumerate}
\end{itemize}    

We consider the following optimization problem.
\[
\textbf{Prob}_1(\text{LC}) := \inf\{\, s\in \{0,1,2, \dots\}:\ \mathcal{F}(s)\neq\emptyset \,\}.
\]
As we will see, $\textbf{Prob}_1(\text{LC})$ is the smallest value of $s$ for which our
precoding problem is feasible. 
\end{problem}

\begin{theorem}\label{theorem:first_main_result_informal} Suppose that the LC problem  satisfies  Assumption~\ref{assumption:lin_ind_assumption}.
\begin{itemize}
    \item[(a)] If an LC problem $(\bbF_q, k, \bfV_1,\dots,\bfV_k)$ is Interval-$[0,1]$, then there exists a zero-error coding scheme over the noiseless quantum many-to-one model with normalized total cost  
\begin{equation}\label{eq:total_cost_Interval-$[0,1]$}
    \Gamma =  \frac{s}{2}+ \frac{\sum_{i\in [k]}m_i}{2},
\end{equation}
where $s=\textbf{Prob}_1(\text{LC})$.
    \item[(b)]  If an LC problem $(\bbF_q, k, \bfV_1,\dots,\bfV_k)$ is Restricted Interval-$[0,1]$, then there exists a zero-error coding scheme over the noiseless quantum many-to-one model with normalized total cost 
\begin{equation}\label{eq:total_cost_restricted_Interval-$[0,1]$}
    \Gamma = \frac{2}{3}\sum_{i\in[k]}m_i.
\end{equation}
\end{itemize}
\end{theorem}
 
\subsubsection*{Discussion of Theorem~\ref{theorem:first_main_result_informal}}

\begin{enumerate}[
    label=\textbf{D\arabic*.},
    ref=D\arabic*,
    leftmargin=3.2em,
    labelsep=0.6em,
    itemsep=0.8em
]

\item \label{disc:thm1-precoding-parameter}
\textbf{Role of the precoding parameter.}
The quantity $s$ measures the minimum auxiliary
dimension needed to transform the two-instance LC problem into a $(\kappa,N)$-Sum Box
instance satisfying the SSO condition. 

\item \label{disc:thm1-prob1-bound} \textbf{Universal bound on $\textbf{Prob}_1$.}
A simple construction gives $\textbf{Prob}_1(\text{LC})\le m$; see Section \ref{sec:preliminaries}.

\item \label{disc:thm1-prior-precoding}
\textbf{Relation to prior precoding approaches.}
 The Optimal Precoding Problem discussed above bears similarity to the optimization problem presented in  Section III of \cite{HuUlukus25}. However, it is different in an important way: the precoding matrix for Alice$_1$ is of different dimension. As demonstrated in Section \ref{sec:comparison_with_prior_works}, this additional freedom in designing the precoding allows us to obtain lower cost schemes than the ones in \cite{HuUlukus25}.  We point out that it is unclear if the above problem can be solved in an efficient manner for arbitrary LC problems; the same issue exists with the work of \cite{HuUlukus25}. 

\item \label{disc:thm1-incomparable-bounds}
\textbf{Incomparability of the two achievable bounds.} It can be observed that depending on the value of $s$ and $\sum_{i \in [k]} m_i$ either of the rates in  \eqref{eq:total_cost_Interval-$[0,1]$} or \eqref{eq:total_cost_restricted_Interval-$[0,1]$} can be lower, i.e., these rates are incomparable in general.

\item \label{disc:thm1-restricted-genericity}
\textbf{Genericity of the Restricted Interval-$[0,1]$ condition.} Suppose that Assumption~\ref{assumption:lin_ind_assumption} holds. Under the uniform sampling model specified in Proposition~\ref{prop:Restricted Interval-$[0,1]$_two_disjoint_submatrices_condition}, an Interval-$[0,1]$ LC problem is Restricted Interval-$[0,1]$ with probability tending to one as $q\to\infty$; see Section~\ref{sec:main_contribution_formal_statement_and_proof}. Thus, the second part of Theorem~\ref{theorem:first_main_result_informal} applies to a generic subclass of Interval-$[0,1]$ LC problems over sufficiently large fields.

\item \label{disc:thm1-classical-comparison}
\textbf{Comparison with the classical many-to-one baseline.}
Combining Theorem~\ref{theorem:first_main_result_informal} with
Proposition~\ref{prop:classical_many_to_one_baseline}, the achievable cost in
Theorem~\ref{theorem:first_main_result_informal}(b) is strictly smaller than
the classical many-to-one optimum whenever $\sum_{i\in[k]}m_i>0$, since
\[
    \frac{2}{3}\sum_{i\in[k]}m_i
    <
    \sum_{i\in[k]}m_i
    =
    \Gamma^{*,\mathsf C}(\mathrm{LC}).
\]
For the bound in part (a), the comparison depends on
$\textbf{Prob}_1(\mathrm{LC})$; using
$\textbf{Prob}_1(\mathrm{LC})\le m$ and the Interval-$[0,1]$ condition
$2m\le\sum_{i\in[k]}m_i$, we obtain
\[
    \frac{1}{2}\textbf{Prob}_1(\mathrm{LC})
    +
    \frac{1}{2}\sum_{i\in[k]}m_i
    \le
    \frac{1}{2}m+\frac{1}{2}\sum_{i\in[k]}m_i
    \le
    \frac{3}{4}\sum_{i\in[k]}m_i
    <
    \sum_{i\in[k]}m_i.
\]
Thus both achievable bounds give a strict improvement over the classical many-to-one optimum for nontrivial instances in their respective regimes.
\item \label{disc:thm1-beyond-interval}
\textbf{Applicability beyond the Interval-$[0,1]$ regime.}
  Theorem \ref{theorem:first_main_result_informal} applies to all LC problems satisfying  $\frac{2m}{\sum_{i\in[k]} m_i}  \in[0,1]$. However, for problems with special structure, it is possible to apply these ideas even when $\frac{2m}{\sum_{i\in[k]} m_i}  \in[1,2]$. In fact, the third coding scheme we propose (Theorem \ref{theorem:direct_sum_coding_formal})  applies when $\frac{2m}{\sum_{i\in[k]} m_i}=\frac{4}{3}$.

\item \label{disc:thm1-limited-room-near-two}
\textbf{Limited room for gain as $\frac{2m}{\sum_i m_i}\to 2$.} As $\frac{2m}{\sum_{i\in[k]}m_i}$ approaches $2$, the potential improvement over the classical many-to-one baseline becomes limited. Indeed, under Assumption~\ref{assumption:lin_ind_assumption}, the desired output can take $q^m$ possible values. Since Bob shares no entanglement with the senders, the quantum converse bound in Lemma~\ref{lemma:quantum_many_to_one_lower_bound} implies that any zero-error quantum many-to-one scheme must satisfy $\Gamma(\Delta)\ge m.$  On the other hand, the classical many-to-one baseline is $\Gamma_{\mathsf C}^*(\mathrm{LC}) = \sum_{i\in[k]}m_i$  by Proposition~\ref{prop:classical_many_to_one_baseline}. Hence any quantum scheme that strictly improves over the classical baseline must satisfy $m \le \Gamma(\Delta) < \sum_{i\in[k]}m_i.$ As $\frac{2m}{\sum_{i\in[k]}m_i}\to 2$, we have $\sum_{i\in[k]}m_i\to m$, so this interval shrinks. Thus there is little room for a distributed-superdense-coding advantage using quantum protocols in this limit.
\end{enumerate}

\subsection{Direct-Sum LC Problems: Subadditivity of $\Gamma^{*,\mathsf Q}$}
\label{subsec:main_contribution_direct-sum}

We next consider joint computation of two distinct LC problems. Throughout this subsection, ``distinct'' means
\[
    (\bfV_1,\dots,\bfV_k)\neq(\bfV_1',\dots,\bfV_k')
\]
as a tuple of  matrices, not merely that the data streams are different. If
$\mathrm{LC}_1$ and $\mathrm{LC}_2$ are computed separately, the total quantum
cost is $\Gamma^{*,\mathsf Q}_{\mathrm{LC}_1} + \Gamma^{*,\mathsf Q}_{\mathrm{LC}_2}.$ However, when the two problems are bundled into a single direct-sum LC problem, joint coding may strictly reduce the total cost. Our second main result shows that this phenomenon indeed occurs under the noiseless quantum many-to-one model.  We first demonstrate this result by means of an example. Next, we generalize this example to a coding scheme that applies to a broad class of $(\text{LC}_1,\text{LC}_2)$.




\begin{definition}[Direct sum of LC problems]
Let $\mathrm{LC}_1=(\bbF_q,k,\bfV_1,\dots,\bfV_k),\,
    \mathrm{LC}_2=(\bbF_q,k,\bfV_1',\dots,\bfV_k')$ be two LC problems over the same field and with the same number of senders. The two problems use disjoint data vectors
$\{W_i\}_{i=1}^k$ and $\{W_i'\}_{i=1}^k$ and compute $Y=\sum_{i\in[k]}\bfV_iW_i,\,Y'=\sum_{i\in[k]}\bfV_i'W_i'.$ Their direct sum is the LC problem
\[
    \mathrm{LC}_1\oplus\mathrm{LC}_2
    :=
    \left(
    \bbF_q,k,
    \begin{bmatrix}\bfV_1&\\&\bfV_1'\end{bmatrix},
    \dots,
    \begin{bmatrix}\bfV_k&\\&\bfV_k'\end{bmatrix}
    \right),
\]
whose output is
\[
    \begin{bmatrix}Y\\Y'\end{bmatrix}
    =
    \sum_{i\in[k]}
    \begin{bmatrix}\bfV_i&\\&\bfV_i'\end{bmatrix}
    \begin{bmatrix}W_i\\W_i'\end{bmatrix}.
\]
\end{definition}
 

\begin{theorem}[Subadditivity]\label{theorem:subadditivity}
    There exist two distinct LC problems $\text{LC}_1=(\bbF_q, k, \bfV_1,\dots,\bfV_k)$ and $\text{LC}_2=(\bbF_q, k, \bfV_1',\dots,\bfV_k')$  over the same field $\bbF_q$ and with the same number of senders $k$, such that
    \begin{align*}
        \Gamma^{*,\mathsf{Q}}_{\text{LC}_1\oplus \text{LC}_2} < \Gamma^{*,\mathsf{Q}}_{\text{LC}_1}+\Gamma^{*,\mathsf{Q}}_{\text{LC}_2}.
    \end{align*} 
\end{theorem}
\begin{remark}
Theorem~\ref{theorem:subadditivity} demonstrates strict subadditivity of the
optimal asymptotic per-instance total cost when two genuinely different LC
problems are combined into a single direct-sum LC problem. This is different
from the multi-letter gain in Example~\ref{example:motivating_example3}, where
the saving comes from jointly coding multiple instances of the same LC problem.
In contrast, under the classical many-to-one model, the optimal total cost is
additive; see Proposition~\ref{prop:classical_many_to_one_additive}.
\end{remark}

Here, we prove Theorem \ref{theorem:subadditivity} by constructing an explicit
pair of LC problems. We leverage an  example from  \cite{Inverted_3_Sum_Box_JafarYao25} (\textbf{Toy Example 3}), which is stated as follows.
\begin{example}\label{example:subadditivity}
Consider the problem $\text{LC}_1=(\bbF_3,3, \bfV_1,\bfV_2,\bfV_3)$ with three senders, where
\begin{align*}
    \bfV_1 = \begin{bmatrix}
        1  \\
        0 
    \end{bmatrix}, \bfV_2 = \begin{bmatrix}
        1 \\
        0 
    \end{bmatrix}, \text{~and~}\bfV_3 =\begin{bmatrix}
         1 &0 \\
         0 &1 
    \end{bmatrix}.
\end{align*}

\end{example} 

\begin{lemma}
(\textbf{Toy Example 3} of \cite{Inverted_3_Sum_Box_JafarYao25}) For any coding scheme of   $\text{LC}_1$, its cost tuple $\Delta=(\Delta_1,\Delta_2,\Delta_3)$ must satisfy
\begin{equation}
    \begin{split}
    &\Delta_1\ge \frac{1}{2},\,\Delta_2\ge \frac{1}{2},\,\Delta_3\ge 1, \\
    &\text{ and } \Delta_1+\Delta_2+\Delta_3\ge \frac{5}{2}.
    \end{split}
\end{equation}
Therefore, the optimal total cost $\Gamma^{*,\mathsf{Q}}_{\text{LC}_1}$ of $\text{LC}_1$ must be such that
\begin{equation}\label{eq:subadditivity_eq1}
    \Gamma^{*,\mathsf{Q}}_{\text{LC}_1}  \ge \frac{5}{2}.
\end{equation}
\end{lemma}
\begin{example}\label{example:subadditivity1}
We define our $\text{LC}_2$ as follows. Alice$_1$ and Alice$_2$ have classical data vectors $W_1'=[x_5]$ and $W_2'=[x_6]$ respectively, and Bob wants to recover   $\begin{bmatrix}
    x_5\\ x_6
\end{bmatrix}$. 
In our terminology, it can be written as $\text{LC}_2=(\bbF_3,3, \bfV_1',\bfV_2',\bfV_3')$ where 
\begin{align*}
    \bfV_1' = \begin{bmatrix}
        1  \\
        0 
    \end{bmatrix}, \bfV_2' = \begin{bmatrix}
        0 \\
        1 
    \end{bmatrix}, \bfV_3'\in\bbF_3^{2\times 0}.
\end{align*} 
Here Alice$_3$ has no input in $\mathrm{LC}_2$; equivalently, $m_3'=0$ and $\bfV_3'$ is the $2\times 0$ empty matrix. For any \(L\)-instance zero-error scheme, Bob must distinguish \(3^{2L}\) possible outputs. Note that zero-error distinguishability requires that Bob receives $3^{2L}$ orthogonal states.  Since the total received Hilbert-space dimension is \(3^{\sum_i\delta_i}\), this implies that \(\sum_i\delta_i\ge 2L\). Hence \(\Gamma^{*,\mathsf Q}_{\mathrm{LC}_2}\ge 2\). The matching upper bound is obtained by Alice$_1$ and Alice$_2$ sending \(W_1'^{[L]}\) and \(W_2'^{[L]}\), respectively.  Therefore, we have 
\begin{equation}\label{eq:subadditivity_eq2}
    \Gamma^{*,\mathsf{Q}}_{\text{LC}_2} = 2.
\end{equation}
\end{example}
\begin{example}\label{example:subadditivity2}
    Now, we consider the case of $\text{LC}_1\oplus \text{LC}_2  =\left(\bbF_3, 3, \begin{bmatrix}
        \bfV_1 & \\ & \bfV_1'
    \end{bmatrix}, \begin{bmatrix}
        \bfV_2 & \\ & \bfV_2'
    \end{bmatrix}, \begin{bmatrix}
        \bfV_3 & \\ & \bfV_3'
    \end{bmatrix}\right)$ where
    \begin{align*}
        \begin{bmatrix}
        \bfV_1 & \\ & \bfV_1'
    \end{bmatrix} =\begin{bmatrix}
        1&0\\
        0&0\\
        0&1\\
        0&0
    \end{bmatrix},
        \begin{bmatrix}
        \bfV_2 & \\ & \bfV_2'
    \end{bmatrix} =\begin{bmatrix}
        1&0\\
        0&0\\
        0&0\\
        0&1
    \end{bmatrix},
        \begin{bmatrix}
        \bfV_3 & \\ & \bfV_3'
    \end{bmatrix} =\begin{bmatrix}
        1&0\\
        0&1\\
        0&0\\
        0&0
    \end{bmatrix}.
    \end{align*}
Let Alice$_1$, Alice$_2$, Alice$_3$ have $\begin{bmatrix}
    x_1\\x_5
\end{bmatrix}, \begin{bmatrix}
    x_2\\x_6
\end{bmatrix}, \begin{bmatrix}
    x_3\\x_4
\end{bmatrix}$ respectively. Then, Bob wants to compute
\begin{align*}
    \begin{bmatrix}
        x_1+x_2+x_3\\ x_4\\ x_5\\x_6
    \end{bmatrix}&= 
    \begin{bmatrix}
        1&0\\
        0&0\\
        0&1\\
        0&0
    \end{bmatrix}\begin{bmatrix}
        x_1\\x_5
    \end{bmatrix} + \begin{bmatrix}
        1&0\\
        0&0\\
        0&0\\
        0&1
    \end{bmatrix} \begin{bmatrix}
        x_2\\ x_6
    \end{bmatrix} + \begin{bmatrix}
        1&0\\
        0&1\\
        0&0\\
        0&0
    \end{bmatrix}\begin{bmatrix}
        x_3\\x_4
    \end{bmatrix}\\
   &= \begin{bmatrix}
        1 & 1 & 1 & 0 \\
        0 & 0 & 0 & 1 \\
        0 & 0 & 0 & 0 \\
        0 & 0 & 0 & 0
    \end{bmatrix} \begin{bmatrix}
        x_1\\x_2\\x_3\\x_4
    \end{bmatrix} + 
    \begin{bmatrix}
      0 & 0 & 0 & 0 \\
      0 & 0 & 0 & 0 \\
      1 & 0 & 2 & 0 \\
      0 & 1 & 2 & 0 
  \end{bmatrix}\begin{bmatrix}
      x_5\\x_6\\0\\0
  \end{bmatrix}.
\end{align*}
Set
\begin{align*}
  M_x =  \begin{bmatrix}
        1 & 1 & 1 & 0 \\
        0 & 0 & 0 & 1 \\
        0 & 0 & 0 & 0 \\
        0 & 0 & 0 & 0
    \end{bmatrix}, \text{~and~}
  M_z = \begin{bmatrix}
      0 & 0 & 0 & 0 \\
      0 & 0 & 0 & 0 \\
      1 & 0 & 2 & 0 \\
      0 & 1 & 2 & 0 
  \end{bmatrix}.    
\end{align*}
Then, we have
\begin{align*}
    &M_xM_z^T-M_zM_x^T\\
    =& 
    \begin{bmatrix}
        1 & 1 & 1 & 0 \\
        0 & 0 & 0 & 1 \\
        0 & 0 & 0 & 0 \\
        0 & 0 & 0 & 0
    \end{bmatrix}   \begin{bmatrix}
      0 & 0 & 1 & 0 \\
      0 & 0 & 0 & 1 \\
      0 & 0 & 2 & 2 \\
      0 & 0 & 0 & 0 
  \end{bmatrix}  -  \begin{bmatrix}
      0 & 0 & 0 & 0 \\
      0 & 0 & 0 & 0 \\
      1 & 0 & 2 & 0 \\
      0 & 1 & 2 & 0 
  \end{bmatrix}   \begin{bmatrix}
        1 & 0 & 0 & 0 \\
        1 & 0 & 0 & 0 \\
        1 & 0 & 0 & 0 \\
        0 & 1 & 0 & 0
    \end{bmatrix} 
    = 0
\end{align*}
and $\text{rank}([M_x|M_z]) = 4$.  Therefore,  there exists a $4$-Sum Box protocol for  computing   $y=M_x\bfx+M_z\bfz$ above. Let $\calQ = \calQ_1\dots \calQ_4$ be the joint system of this $4$-Sum Box. Alice$_1$, Alice$_2$ have $\calQ_1$ and $\calQ_2$, respectively, and Alice$_3$ has $\calQ_3$ and $\calQ_4$. Then
\begin{itemize}
    \item Alice$_1$ encodes $ [x_1, x_5]$ to $\calQ_1$. 
    \item Alice$_2$ encodes $[x_2,x_6]$ to $\calQ_2$.
    \item Alice$_3$ encodes $[x_3,0] $ to $\calQ_3$.
    \item Alice$_3$ encodes $[x_4,0] $ to $\calQ_4$.
\end{itemize}  Then, they run the $4$-Sum Box protocol and $[x_1+x_2+x_3~ x_4~ x_5~x_6]^T$ is computed.
    In this 4-Sum Box protocol, Alice$_1$ and Alice$_2$ transmit one qudit, and Alice$_3$ transmits two qudits. Therefore, the cost tuple is $\Delta = (1,1,2)$, and the total cost is $\Gamma(\Delta) = 1+1+2=4$. Recall that \eqref{eq:subadditivity_eq1} and \eqref{eq:subadditivity_eq2} give $ \Gamma^{*,\mathsf{Q}}_{\text{LC}_1 }\ge 2.5$ and $\Gamma^{*,\mathsf{Q}}_{  \text{LC}_2}=2$ respectively. Therefore, we have that 
    \begin{align*}
        \Gamma^{*,\mathsf{Q}}_{\text{LC}_1\oplus \text{LC}_2} \le \Gamma(\Delta)= 4 <\frac{9}{2} \le \Gamma^{*,\mathsf{Q}}_{\text{LC}_1 } + \Gamma^{*,\mathsf{Q}}_{  \text{LC}_2}.
    \end{align*}
This proves Theorem \ref{theorem:subadditivity}. A generalization of  this example appears in Section \ref{subsec:direct_sum_coding_scheme}. 

We emphasize that the strict subadditivity of $\Gamma^{*,Q}$ in this example arises from combining two distinct LC problems into one LC problem, whereas Example \ref{example:motivating_example3} demonstrates multi-letter savings for a single LC problem.    
\end{example}
  


\section{Preliminaries}\label{sec:preliminaries}
In this section, we present basic definitions that will be used in the sequel. We also discuss the basic result of the $N$-Sum Box paper \cite{N_Sum_Box_JafarYao25}.
Many of our proofs will involve block matrices that contain all-zero blocks with different dimensions. We use subscripts to specify the dimensions of the zero block matrices under consideration, e.g., $0_{a\times b}$ denotes the $0$ matrix of dimension $a\times b$.  We omit the subscript whenever the dimensions are apparent. 


\subsection{Linear Algebra Facts}

\begin{definition}[Alternating Matrix]\label{def:alternating_matrix}
    Let $M=(m_{i,j})_{i,j=1}^n$ be an $n\times n$ matrix. $M$ is said to be alternating if $m_{i,j} = -m_{j,i}$ and $m_{i,i}=0$ for all $i,j\in[n]$. It can be shown that $\text{rank}(M)$ is always even.\footnote{When $\text{char}(\bbF_q)=2$, the condition $M^{T}=-M$ means $M$ is symmetric, so the extra condition $m_{i,i}=0$ is essential.}
\end{definition} 
\begin{proposition}

\label{prop:standard_form_of_mtx:alternating}  (Equation (5.3) of \cite{ConradBilinearForms}). 
When $M$ is alternating and has dimension $2n\times 2n$, there exists a $2n\times 2n$ invertible matrix $U$ such that\footnote{Here, we use the notation $U M U^T$ instead of $U^T M U$ because it will be more convenient for later use.}
        \begin{align*}
            U M U^T = \begin{bmatrix}
                0_{n\times n} & -D_r\\
                D_r^T& 0_{n\times n}
            \end{bmatrix},
        \end{align*}
        where $2r = \text{rank}(M)$ and $D_r = \begin{bmatrix}
                    I_{r} & 0_{r\times (n-r)}\\
                    0_{(n-r)\times r} & 0_{(n-r)\times (n-r)}
                \end{bmatrix}$. 
                
\end{proposition} 

\begin{lemma} 
\label{lem:alternating_omega_factorization}
Let \(M\in\bbF_q^{2m\times 2m}\) be alternating with \(\text{rank}(M)=2r\). Then, there exist matrices \(Q,R\in\bbF_q^{2m\times r}\) such that $M=\Omega(Q,R)$ (\textit{cf.} Definition \ref{def:SSO_matrix} for $\Omega(Q,R)$).  
\end{lemma}

\begin{proof}
By Proposition \ref{prop:standard_form_of_mtx:alternating}, there exists an
invertible matrix \(S\) such that
\[
SMS^T = 
\begin{bmatrix}
0 & -D_r\\
D_r^T & 0
\end{bmatrix},
\qquad
D_r=
\begin{bmatrix}
I_r & 0\\
0 & 0
\end{bmatrix}
\in\bbF_q^{m\times m}.
\]
Let $U = S^{-1}$, and \(E_r=\begin{bmatrix}I_r\\0\end{bmatrix}\in\bbF_q^{m\times r}\), so
\(D_r=E_rE_r^T\). Since $
\Omega\left(
\begin{bmatrix}0\\E_r\end{bmatrix},
\begin{bmatrix}E_r\\0\end{bmatrix}
\right)
=
\begin{bmatrix}
0 & -D_r\\
D_r^T & 0
\end{bmatrix},$
the claim follows by taking $
Q=U\begin{bmatrix}0\\E_r\end{bmatrix},
\qquad
R=U\begin{bmatrix}E_r\\0\end{bmatrix}.$
\end{proof}

   
The following proposition is standard and follows from a simple counting argument;
see, e.g., \cite{Stanley_2011}.
\begin{proposition}\label{prop:number_of_FullCol} 
 Let $\#\text{FullCol}(m,n;q)$ denote the number of $m\times n$ matrices over $\bbF_q$ with full column rank. If $0< n \le m$,  $\#\text{FullCol}(m,n;q) =(q^m-1)(q^m-q)\dots(q^m-q^{n-1}).$ 
\end{proposition} 
\subsection{$(\kappa,N)$-Sum Box protocol}\label{subsec:kappa_n-sum_box}

Throughout this subsection, let $M_x$ and $M_z$ be of dimension $\kappa \times N$ so that the block matrix $M=[M_x|M_z]\in\bbF_q^{\kappa\times 2N}$. The following proposition about the SSO condition matrix ({\it cf.} Definition \ref{def:SSO_matrix}) is used extensively in the sequel. 
\begin{proposition}\label{prop:SSO_proposition} 
Let all matrices below have compatible dimensions. 
\begin{enumerate}[label=(\alph*)]
    \item \label{prop:SSO_partition}
    Let  $M_{x_i},M_{z_i}$ be $\kappa \times N_i$ matrices over $\bbF_q$ for  $i=1,2$.
    Then,
\begin{align*}
        \Omega([M_{x_1}|M_{x_2}],[M_{z_1}|M_{z_2}]) &= \Omega(M_{x_1},M_{z_1} ) + \Omega(M_{x_2},M_{z_2}).
\end{align*}
\item \label{prop:SSO_asymmetric}  $\Omega(M_x,M_z)= -\Omega(M_z,M_x)$.  
\item  \label{prop:SSO_bilinear}
    $\Omega(M_x,M_z)=M_{x}M_{z}^T- M_{z}M_{x}^T$ is bilinear.  
    \item  \label{prop:SSO_congruence} Let $D$ be a $\kappa\times \kappa$ matrix over $\bbF_q$. Then  $\Omega(DM_x,DM_z)=D\Omega(M_x,M_z)D^T.$ 
\item  $\Omega(M_x,M_z)$ is alternating regardless of the characteristic of $\bbF_q$. \label{prop:SSO_alternating} 
\end{enumerate}  
\end{proposition} 
\begin{proof}
    See Appendix~\ref{prop:SSO_proposition:app}.
\end{proof}
The functionality of the $(\kappa,N)$-Sum Box is as follows. Let $q=p^e$ where $p$ is prime and $e\in \bbN$. Suppose a $q$-dimensional quantum system $\calH$ is spanned by an orthonormal basis $\{\ket{j}:j\in\bbF_q\}$. For $x\in \bbF_q$, we define the trace $\text{tr}_{\bbF_q/\bbF_p}(x):= \sum_{\ell=0}^{e-1}x^{p^\ell}\in\bbF_p$. Let $\omega := \text{exp}(2\pi \textrm{i} /p)$ where $\textrm{i} = \sqrt{-1}$. For $a,b\in\bbF_q$, we define unitary matrices 
$\mathscr{X}(a):= \sum_{j\in\bbF_q}\ket{j+a}\bra{j}$ and $\mathscr{Z}(b):= \sum_{j\in\bbF_q}\omega^{\text{tr}_{\bbF_q/\bbF_p}(bj)}\ket{j}\bra{j}$. 
\begin{lemma}[Operational form of the $N$-Sum Box~\cite{N_Sum_Box_JafarYao25}]
\label{lemma:N_sum_box_operational}
Let $M_x,M_z\in\bbF_q^{N\times N}$ and
$M=[M_x|M_z]\in\bbF_q^{N\times 2N}$. Suppose that
\begin{enumerate}[label=(\arabic*)]
    \item $\Omega(M_x,M_z)=0$, and
    \item $\operatorname{rank}(M)=N$.
\end{enumerate}
Then there exist mutually orthogonal density operators $\{\rho_{\bfv}^{M}:\bfv\in\bbF_q^N\}$ on the composite system $\calQ=\calQ_1\cdots\calQ_N$, where each $\calQ_i$ is a $q$-dimensional subsystem, such that for all $\bfx,\bfz\in\bbF_q^N$ and all $\bfv\in\bbF_q^N$,
\[
    \left(
        \bigotimes_{i=1}^N
        \mathscr{X}(x_i)\mathscr{Z}(z_i)
    \right)
    \rho_{\bfv}^{M}
    \left(
        \bigotimes_{i=1}^N
        \mathscr{X}(x_i)\mathscr{Z}(z_i)
    \right)^\dagger
    =
    \rho_{\bfv'}^{M},
\]
where $\bfv'=\bfv+M_x\bfx+M_z\bfz\in\bbF_q^N.$  Here the superscript $M$ indicates that the orthogonal family of states depends on the matrix $M$.
\end{lemma}

Lemma~\ref{lemma:N_sum_box_operational} is the $N$-Sum Box primitive used as a black box in this paper. 

 \begin{remark}[Locally invertible transforms for the $N$-Sum Box]
\label{remark:LIT_N_sum_box}
The $N$-Sum Box black box also has a locally invertible transform (LIT)
functionality characterized in~\cite{N_Sum_Box_JafarYao25}. Let
\[
    \Lambda=\operatorname{diag}(\lambda_1,\dots,\lambda_N),
    \qquad \lambda_i\in\bbF_q \setminus \{0\},
\]
be an invertible diagonal matrix, and let $M=[M_x|M_z]\in\bbF_q^{N\times 2N}$. Suppose that $\operatorname{rank}([M_x|M_z])=N$ and that the transformed pair $(M_x,M_z\Lambda)$ satisfies $\Omega(M_x,M_z\Lambda)=0.$ Then Lemma~\ref{lemma:N_sum_box_operational} applies to $M^{\Lambda}:=[M_x|M_z\Lambda].$ Operationally, this Sum Box for $M^\Lambda$ can be used to compute the original transfer matrix $M=[M_x|M_z]$: Alice$_i$ locally replaces $z_i$ by $\lambda_i^{-1}z_i$ before applying the Weyl operator. Bob then obtains $M_x\bfx+M_z\Lambda(\Lambda^{-1}\bfz) = M_x\bfx+M_z\bfz.$ When $q=2$, the only nonzero field element is $1$, so every invertible diagonal matrix is the identity, i.e., the usage of LITs need not be considered when the computation is over $\bbF_2$. We point out that all our proposed coding schemes in this paper use the identity LIT, i.e., $\Lambda=I_N$.
\end{remark}


\begin{lemma}[Operational form of the $(\kappa,N)$-Sum Box~\cite{N_Sum_Box_JafarYao25}] 
\label{lemma:kappa,n-sum_box} 
Let $M_x,M_z\in\mathbb F_q^{\kappa\times N}$ and
$M=[M_x\mid M_z]$. Suppose that
\begin{enumerate}[label=(\arabic*)]
    \item $\Omega(M_x,M_z)=0$;
    \item $\text{rank}(M)=\kappa$;
    \item $\kappa\le N$.
\end{enumerate}
Then there exist matrices
$R_x,R_z\in\mathbb F_q^{(N-\kappa)\times N}$ such that
\[
\widehat M_x=
\begin{bmatrix}
M_x\\ R_x
\end{bmatrix},
\qquad
\widehat M_z=
\begin{bmatrix}
M_z\\ R_z
\end{bmatrix}
\]
satisfy
\[
\Omega(\widehat M_x,\widehat M_z)=0,
\qquad
\text{rank}([\widehat M_x\mid\widehat M_z])=N.
\]
Consequently, the $N$-Sum Box protocol associated with
$(\widehat M_x,\widehat M_z)$ allows Bob to recover
\[
\widehat M_xx+\widehat M_zz
=
\begin{bmatrix}
M_xx+M_zz\\
R_xx+R_zz
\end{bmatrix}.
\]
By retaining the first $\kappa$ coordinates and discarding the
remaining $N-\kappa$ coordinates, Bob recovers
$M_xx+M_zz$ exactly using $N$ transmitted qudits. We refer to
this completion-and-discard procedure as a $(\kappa,N)$-Sum Box
protocol.
\end{lemma}


\begin{proof}
Let $J=\begin{bmatrix}
        0&I_N\\
        -I_N&0
    \end{bmatrix}, \,
    G=M^T\in\bbF_q^{2N\times\kappa}.$ The assumptions imply $G^TJG=MJM^T=\Omega(M_x,M_z)=0,\,     \operatorname{rank}(G)=\kappa.$ By Lemma~1 of~\cite{N_Sum_Box_JafarYao25}, $G$ can be completed to a rank-$N$ strongly self-orthogonal matrix $\widehat G=[G|G_{\mathrm{aux}}]\in\bbF_q^{2N\times N},\,
    \widehat G^TJ\widehat G=0.$ Set $\widehat M=\widehat G^T$. Since the first $\kappa$ rows of $\widehat M$
are the rows of $M$, we can write
\[
    \widehat M
    =
    \begin{bmatrix}
        M\\ R
    \end{bmatrix}
    =
    [\widehat M_x|\widehat M_z],
    \qquad
    \widehat M_x=
    \begin{bmatrix}
        M_x\\R_x
    \end{bmatrix},
    \quad
    \widehat M_z=
    \begin{bmatrix}
        M_z\\R_z
    \end{bmatrix}.
\]
Then $\Omega(\widehat M_x,\widehat M_z)=0,\,
    \operatorname{rank}(\widehat M)=N.$ 
Hence the $N$-Sum Box (\textit{cf.} Lemma \ref{lemma:N_sum_box_operational}) applies to $(\widehat M_x,\widehat M_z)$ and Bob obtains
\[
    \widehat M_x\bfx+\widehat M_z\bfz
    =
    \begin{bmatrix}
        M_x\bfx+M_z\bfz\\
        R_x\bfx+R_z\bfz
    \end{bmatrix}.
\]
Bob keeps the first $\kappa$ coordinates and discards the remaining
$N-\kappa$ coordinates. In this paper we use this construction with identity
LIT.
\end{proof}  
Operationally, the senders run the $N$-Sum Box associated with
$(\widehat M_x,\widehat M_z)$. Alice$_i$ applies
$X(x_i)Z(z_i)$ to her subsystem $Q_i$ and transmits $Q_i$ to
Bob. Bob obtains the $N$-dimensional output
\[
\widehat M_xx+\widehat M_zz
=
\begin{bmatrix}
M_xx+M_zz\\
R_xx+R_zz
\end{bmatrix},
\]
keeps its first $\kappa$ coordinates, and discards the remaining $N-\kappa$ auxiliary coordinates. Initially, the senders share the state $\rho_{\mathbf{0}}^M$, with Alice$_i$ holding subsystem $\calQ_i$. Upon receiving input pair $(x_i,z_i)$, Alice$_i$ applies the local Weyl operator $\mathscr{X}(x_i)\mathscr{Z}(z_i)$ to $\calQ_i$ and then sends $\calQ_i$ to Bob through her noiseless quantum link. By Lemma~\ref{lemma:kappa,n-sum_box}, the joint state received by Bob is $\rho_{\bfv'}^{M} $  where $\bfv'  = M_x\bfx+M_z\bfz\in \bbF_q^{\kappa}$. Since the states $\{\rho_{\bfv}^{M}:\bfv\in\bbF_q^\kappa\}$ are mutually orthogonal, Bob can distinguish them with zero error (see Lemma 1.2 of \cite{BrietBLPS15}). 
Hence the measurement outcome is $M_x\bfx+M_z\bfz\in \bbF_q^{\kappa}$,  which is the desired linear function.

\begin{proposition} \label{prop:trivial_precoding_bound}
Let \(\mathrm{LC}=(\bbF_q,k,\bfV_1,\dots,\bfV_k)\) be an Interval-\([0,1]\) problem. For invertible matrices \(P_i\in\bbF_q^{m_i\times m_i}\), define
\[
M(P_1,\dots,P_k)
:=
\sum_{i\in[k]}
\Omega\left( 
        \begin{bmatrix}
            \bfV_iP_i\\0
        \end{bmatrix},
        \begin{bmatrix}
            0\\ \bfV_i
        \end{bmatrix}
     \right)
\in\bbF_q^{2m\times 2m}.
\]
Then
\[
    \mathbf{Prob}_1(\mathrm{LC})
    \le
    \frac{1}{2}\text{rank} (M(P_1,\dots,P_k))
    \le m .
\]
\end{proposition} 
\begin{proof}
    The proof simply cancels the residual SSO obstruction \(M(P_1,\dots,P_k)\) by representing this alternating obstruction as \(\Omega(Q,R)\) using \(r=\text{rank}(M)/2\) auxiliary columns (see Appendix~\ref{app:prop:trivial_precoding_bound} for details). 
\end{proof}

\section{Proofs of Results}\label{sec:main_contribution_formal_statement_and_proof}

We begin with a formal definition and discussion of Restricted Interval-$[0,1]$ LC problems and then present the proofs of the results stated in  Section~\ref{sec:statement_main_contributions}.
\begin{definition}[{Restricted Interval-$[0,1]$ LC}]
\label{def:Restricted Interval-$[0,1]$}
An Interval-$[0,1]$ LC problem $(\bbF_q,k,\bfV_1,\dots,\bfV_k)$ is called Restricted Interval-$[0,1]$ if there are an integer $L\ge 1$ and non-negative integers $a_1,\dots,a_k$ satisfying
\[
    a_i\le \frac{m_iL}{2},\quad \forall i\in[k], \text{~and~}
    \sum_{i\in[k]}a_i=mL,
\]
such that, for each $i\in[k]$, one can choose $a_i$ columns of $\bfV_i^{[L]}$, denoted by the submatrix $A_i$, so that the matrix $[A_1|A_2|\cdots|A_k]$  is invertible.
\end{definition}


\begin{example}[{A Restricted Interval-[0,1] LC problem}]
\label{example:restricted_interval_example}
Consider the LC problem $\mathrm{LC}=(\bbF_2,3,\bfV_1,\bfV_2,\bfV_3)$
with $m=2$ and $m_i=2, i = 1, \dots, 3$ (so that it is an Interval-$[0,1]$ LC problem) and
\[
Y = \underbrace{\begin{bmatrix}
    1&1\\0&1
\end{bmatrix} }_{\bfV_1}  \begin{bmatrix}
    W_{11}\\W_{12}
\end{bmatrix} + \underbrace{\begin{bmatrix}
    1&0\\1&1
\end{bmatrix} }_{\bfV_2}  \begin{bmatrix}
    W_{21}\\W_{22}
\end{bmatrix} + \underbrace{\begin{bmatrix}
    0&1\\1&0
\end{bmatrix} }_{\bfV_3}  \begin{bmatrix}
    W_{31}\\W_{32}
\end{bmatrix}
\] 
We now show that it is Restricted Interval-$[0,1]$. Write the columns as
\[
    \bfV_1=[\alpha_1\mid \beta_1],
    \qquad
    \bfV_2=[\beta_2\mid \alpha_2],
    \qquad
    \bfV_3=[\gamma_1\mid \gamma_2],
\]
where $\alpha_1=\begin{bmatrix}1\\0\end{bmatrix},$ and $\alpha_2=\begin{bmatrix}0\\1\end{bmatrix}.$ Choose $A_1=[\alpha_1]$, $A_2=[\alpha_2]$,  $A_3=\emptyset.$  Then $a_1=a_2=1$, $a_3=0$  and
\[
    a_i\le \frac{m_i}{2}
    \qquad \text{for all } i\in\{1,2,3\}.
\]
Moreover,  $[A_1\mid A_2] =
    \begin{bmatrix}
        1&0\\
        0&1
    \end{bmatrix}$ is invertible. Hence the problem is Restricted Interval-$[0,1]$ with $L=1$.

 
\end{example}

\begin{example}[An Interval-{[0,1]} LC problem that is not Restricted]
\label{example:interval_not_restricted}
Consider
\[
\mathrm{LC}
=
\left(
\bbF_2,4,
\bfV_1=\begin{bmatrix}1\\0\end{bmatrix},
\bfV_2=\begin{bmatrix}0\\1\end{bmatrix},
\bfV_3=\begin{bmatrix}1\\0\end{bmatrix},
\bfV_4=\begin{bmatrix}1\\0\end{bmatrix}
\right).
\]
Thus each Alice$_i$ has one input symbol, and Bob wants to compute
\[
    Y
    =
    \begin{bmatrix}1\\0\end{bmatrix}W_1
    +
    \begin{bmatrix}0\\1\end{bmatrix}W_2
    +
    \begin{bmatrix}1\\0\end{bmatrix}W_3
    +
    \begin{bmatrix}1\\0\end{bmatrix}W_4.
\]
Here $m=2$ and $\sum_{i=1}^4m_i=4$, so $\frac{2m}{\sum_{i=1}^4m_i}=1.$ Hence this is an Interval-$[0,1]$ LC problem. We claim that it is not Restricted Interval-$[0,1]$. Fix any $L\ge 1$. After a common permutation of the $2L$ output coordinates, we may write
\[
    \bfV_2^{[L]}
    =
    \begin{bmatrix}
        0_{L\times L}\\
        I_L
    \end{bmatrix},
    \qquad
    \bfV_1^{[L]}
    =
    \bfV_3^{[L]}
    =
    \bfV_4^{[L]}
    =
    \begin{bmatrix}
        I_L\\
        0_{L\times L}
    \end{bmatrix}.
\]
Therefore,
\[
    \text{col}(\bfV_2^{[L]})
    \cap
    \text{col}\!\left(
    [\bfV_1^{[L]}\mid \bfV_3^{[L]}\mid \bfV_4^{[L]}]
    \right)
    =
    \{0\}.
\]
The columns of $\bfV_2^{[L]}$ are the only columns that have nonzero components
in the last $L$ coordinates. Hence any invertible $2L\times 2L$ submatrix of
\[
    [\bfV_1^{[L]}\mid \bfV_2^{[L]}\mid \bfV_3^{[L]}\mid \bfV_4^{[L]}]
\]
must include all $L$ columns of $\bfV_2^{[L]}$. In the notation of
Definition~\ref{def:Restricted Interval-$[0,1]$}, this forces $a_2=L.$  However, $m_2=1$, so the Restricted condition would require $a_2\le \frac{m_2L}{2}=\frac{L}{2},$ which is impossible. Therefore this LC problem is Interval-$[0,1]$ but not Restricted Interval-$[0,1]$. 
\end{example}
We next introduce a simple sufficient condition for an Interval-$[0,1]$ LC problem to be Restricted. This condition will also be useful for showing that Restricted Interval-$[0,1]$ LC problems form a large class over sufficiently large fields. 
\begin{definition}[Double-Basis LC]\label{def:Double-Basis_LC}
    An LC problem   $ (\bbF_q, k, \bfV_1,\dots,\bfV_k)$ is said to be Double-Basis if the matrix $\bfV = [\bfV_1|\dots|\bfV_k]$ contains two column-disjoint $m\times m$ full-rank submatrices $D,D'$. It can be seen that this condition can hold only if $\sum_{i\in[k]}m_i\ge 2m$.
\end{definition} 
\begin{definition}\label{def:proportion_of_Double-Basis}
Let $k,\bbF_q$ and $m,m_1,\dots,m_k$ be fixed. Let $\text{Num}_A$ denote the number of LC problems $(\bbF_q, k, \bfV_1,\dots,\bfV_k)$ that are (i) Double‑Basis, and (ii) satisfy Assumption \ref{assumption:lin_ind_assumption}, and $\text{Num}_B$ denote the number of LC problems $(\bbF_q, k, \bfV_1,\dots,\bfV_k)$ that satisfy Assumption \ref{assumption:lin_ind_assumption} only.  Define $\rho = \frac{\text{Num}_A}{\text{Num}_B} \leq 1$ to be the proportion of Double-Basis LC problems among all LC problems that satisfy Assumption \ref{assumption:lin_ind_assumption}. Note that $\rho$ depends on $k,q,m,m_1,\dots,m_k$.
\end{definition}
The proof of the following proposition appears in Appendix \ref{app:prop:Restricted Interval-$[0,1]$_two_disjoint_submatrices_condition}.
\begin{proposition}\label{prop:Restricted Interval-$[0,1]$_two_disjoint_submatrices_condition}
Let  $(\bbF_q, k, \bfV_1,\dots,\bfV_k)$  be an LC problem such that $\frac{2m}{\sum_{i\in[k]} m_i}  \in[0,1]$ and  Assumption \ref{assumption:lin_ind_assumption} is satisfied.  The following statements hold.
\begin{enumerate}[label=(\alph*)]
    \item If  $(\bbF_q, k, \bfV_1,\dots,\bfV_k)$ is Double-Basis, then  $(\bbF_q, k, \bfV_1,\dots,\bfV_k)$   is Restricted Interval-$[0,1]$.
    \item There exists an algorithm to decide whether $(\bbF_q, k, \bfV_1,\dots,\bfV_k)$ is Double-Basis in   $O(\text{poly}(m,\sum_{i\in[k]} m_i))$ time.
    \item The fraction $\rho\ge \max\left(1- \frac{k+2}{q-1},0\right)$.
\end{enumerate} 
\end{proposition}   

\begin{remark} 

\begin{itemize} 
\item   Proposition \ref{prop:Restricted Interval-$[0,1]$_two_disjoint_submatrices_condition}(b) shows that one can verify in polynomial time whether a given LC problem is a Double‑Basis LC problem. From Proposition \ref{prop:Restricted Interval-$[0,1]$_two_disjoint_submatrices_condition}(a), it follows that we can apply  the coding scheme of Restricted Interval-$[0,1]$ LC problems to Double-Basis LC problems.
 
\item  For a fixed number of senders, Proposition~\ref{prop:Restricted Interval-$[0,1]$_two_disjoint_submatrices_condition}(c) shows that, when the field size $q$ is sufficiently large, an LC problem sampled uniformly from the class satisfying Assumption~\ref{assumption:lin_ind_assumption} is Double-Basis with high probability. Hence, by part~(a), it is also Restricted Interval-$[0,1]$ with high probability. 
\end{itemize}  
\end{remark}


\subsection{Theorem \ref{theorem:first_main_result_informal}(a): Proof and Discussion} 
The idea is to compute two instances jointly. Alice$_1$ is allowed to use $s$ auxiliary qudits, and the precoding variables $(\tilde{M}_1,\tilde{M}_2,\tilde{X}_1,\{P_i\}_{i=2}^k)$ are chosen so that the two-instance LC problem can be represented as a Sum Box instance satisfying the SSO condition. The parameter $s$ measures the number of auxiliary dimensions needed to make this possible. 
\begin{proof}  
 We consider an $L=2$ instance coding scheme. With the definitions of $P_i,\, i = 2, \dots, k$ and $\tilde{M}_1, \tilde{M}_2$ and $\tilde{X}_1$ that satisfy the Optimal Precoding Problem for a given $s \geq 0$, we observe that

\begin{equation}\label{eq:lemma6_eq4}
    \begin{split}
& \begin{bmatrix}
        \bfV_1&\\
        &\bfV_1
    \end{bmatrix}\begin{bmatrix}
        W_1^{(1)}\\W_1^{(2)}
    \end{bmatrix} + \dots +\begin{bmatrix}
        \bfV_k&\\
        &\bfV_k
    \end{bmatrix}\begin{bmatrix}
        W_k^{(1)}\\W_k^{(2)}
    \end{bmatrix} \\
    =& [\tilde{M}_1 | \tilde{M}_2]\tilde{X}_1 \begin{bmatrix}
        W_1^{(1)}\\W_1^{(2)}
    \end{bmatrix} + \sum_{i\in [k] \setminus \{1\} }  \begin{bmatrix}
        \bfV_i&\\
        &\bfV_i
    \end{bmatrix}\begin{bmatrix}
    P_i&\\
    &I_{m_i}
\end{bmatrix}\begin{bmatrix}
    P_i^{-1}&\\
    &I_{m_i}
\end{bmatrix} \begin{bmatrix}
        W_i^{(1)}\\W_i^{(2)}
    \end{bmatrix} \\ 
    =& [\tilde{M}_1|\tilde{M}_2] \left(\tilde{X}_1 \begin{bmatrix}
        W_1^{(1)}\\W_1^{(2)}
    \end{bmatrix} \right)+ \sum_{i\in [k] \setminus \{1\} }  \begin{bmatrix}
        \bfV_iP_i&\\
        &\bfV_i
    \end{bmatrix} \begin{bmatrix}
       P_i^{-1} W_i^{(1)}\\W_i^{(2)}
    \end{bmatrix}  
    \end{split}
\end{equation} 
where we recall that $\tilde{M}_1,\tilde{M}_2\in\bbF_q^{2m\times(m_1+s)}$, and $\tilde{X}_1\in \bbF_q^{2(m_1+s)\times 2m_1}$.
We define the new data vectors
\begin{equation}\label{eq:def_of_hat_w}
    \begin{bmatrix}
    \hat W_{x,1}\\
    \hat W_{z,1}
\end{bmatrix}  := \tilde{X}_1 \begin{bmatrix}
        W_1^{(1)}\\W_1^{(2)}
    \end{bmatrix} ,\, \begin{bmatrix}
    \hat W_{x,i}\\
    \hat W_{z,i}
\end{bmatrix}  := \begin{bmatrix}
       P_i^{-1} W_i^{(1)}\\W_i^{(2)}
    \end{bmatrix}\text{ for }i\in\{2,\dots,k\},
\end{equation}  where $\hat W_{x,i}$ and $\hat W_{z,i}$ are the first and second $m_i$ components, respectively for $i\in\{2,\dots,k\}$ and $\hat W_{x,1}$ and $\hat W_{z,1}$ have $m_1+s$ components. Then, \eqref{eq:lemma6_eq4} can be written
    as 
\begin{equation}
    \begin{split}\label{eq:lemma6_eq3}
   \begin{bmatrix}
    Y^{(1)}\\Y^{(2)}
\end{bmatrix} =  [\tilde{M}_1|\tilde{M}_2]  \begin{bmatrix}
    \hat W_{x,1}\\
    \hat W_{z,1}
\end{bmatrix} + \sum_{i\in [k] \setminus \{1\} }  \begin{bmatrix}
        \bfV_iP_i&\\
        &\bfV_i
    \end{bmatrix} \begin{bmatrix}
    \hat W_{x,i}\\
    \hat W_{z,i}
\end{bmatrix}  .
    \end{split}
\end{equation} 
Set $M_x = [M_{x,1}|\dots|M_{x,{k}}],$ $M_z = [M_{z,1}|\dots|M_{z,{k}}]$ where
\begin{align*} 
   &M_{x,1} = \tilde{M}_1, M_{z,1} = \tilde{M}_2, \\
   & M_{x,i} = \begin{bmatrix}
        \bfV_iP_{i}\\
        0
    \end{bmatrix}, M_{z,i} = \begin{bmatrix}
        0\\
        \bfV_i
    \end{bmatrix}\text{ for }i\in\{2,\dots,k\}.
\end{align*}  
Let $\kappa = 2m,N = s+\sum_{i=1}^km_i$. We note that $2m\le s+\sum_{i=1}^km_i$ since the LC problem is Interval-$[0,1]$. The constraints of the Optimal Precoding Problem imply that $\Omega(M_x,M_z)=0$.  
To see the rank condition, we note that there exists a block precoding matrix $R$ induced by
$\tilde X_1$ and $\{P_i^{-1}\}_{i=2}^k$ such that 
\[
    [M_x|M_z]R
    =
    [\bfV_1^{[2]}|\bfV_2^{[2]}|\cdots|\bfV_k^{[2]}]
\]
For the sake of completeness, a proof of this statement appears in Appendix~\ref{app:block_precoding_rank_identity}. 
By
Assumption~\ref{assumption:lin_ind_assumption}, the matrix on the right has
rank $2m$. Hence $\text{rank}([M_x|M_z])\ge 2m$. Since $[M_x|M_z]$ has
only $2m$ rows, we conclude that
\[
    \text{rank}([M_x|M_z])=2m.
\]

 Therefore, all three conditions of  Lemma \ref{lemma:kappa,n-sum_box} are satisfied and one can compute  \eqref{eq:lemma6_eq3} 
by a $(2m, s+\sum_{i=1}^km_i )$-Sum Box with identity LIT (\textit{cf.} Remark~\ref{remark:LIT_N_sum_box}). This $(2m,s+\sum_{i=1}^km_i )$-Sum Box is described as follows.  Let $X$ and $Z$ be the vectors obtained by vertically stacking $(\hat W_{x,1},\dots,\hat W_{x,{k}})$ and $(\hat W_{z,1},\dots,\hat W_{z,{k}})$ respectively, i.e., 
\begin{align*}
    X^T = \begin{bmatrix}
        \hat W_{x,1}^T&
        \dots&
        \hat W_{x,{k}}^T
    \end{bmatrix},\, Z^T = \begin{bmatrix}
        \hat W_{z,1}^T&
        \dots&
        \hat W_{z,{k}}^T
    \end{bmatrix}.
\end{align*}
By \eqref{eq:def_of_hat_w}, each $\hat W_{x,i}$ and $\hat W_{z,i}$ can be obtained by local precoding by Alice$_i$. Then, each $X_j$ and $Z_j$ belong to the same Alice for $j\in[s+\sum_{i\in[k]} m_i]$. We assign $\calQ_j$ to the sender who possesses $X_j$ and $Z_j$.  It follows that the number of qudits sent by Alice$_1$ is $s+m_1$ and Alice$_i$ for $i \in \{2, \dots, k\}$ sends $m_i$ qudits.

Recall that $L=2$, so that
\begin{align*}
    \Delta = \bigg{(}\frac{s+m_1}{2},\frac{m_2}{2},\dots,\frac{m_k}{2}\bigg{)}\text{ and }\Gamma(\Delta)=\frac{s+\sum_{i=1}^km_i}{2}.
\end{align*}  
Finally, we can minimize $s$ by considering the solution to $\textbf{Prob}_1(\text{LC})$.
 
\end{proof}   

\subsection{Theorem \ref{theorem:first_main_result_informal}(b): Proof and Discussion}
 Before embarking on the proof of Theorem \ref{theorem:first_main_result_informal}(b), we discuss an example that illustrates the core idea of our proof.

\begin{example}\label{example:lemma_Restricted Interval-$[0,1]$_coding_scheme}
Given an LC problem $(\bbF_q,k,\bfV_1,\dots,\bfV_k)$ with $\frac{2m}{\sum_{i\in[k]}m_i}\in[0,1]$, suppose that
by permuting columns of $[\bfV_1|\dots|\bfV_k]$ and rows of $\begin{bmatrix}
        W_1\\\vdots\\W_k
    \end{bmatrix}$, we can write
\begin{align*}
    Y = &[\bfV_1|\dots|\bfV_k]\begin{bmatrix}
        W_1\\\vdots\\W_k
    \end{bmatrix} =  [I_m| \,\bfP|\,\bfQ] \begin{bmatrix}
        U\\V\\Z
    \end{bmatrix} = U+\bfP V +\bfQ Z,
\end{align*}
where 
\begin{itemize}
    \item  $\begin{bmatrix}
        U\\V\\Z
    \end{bmatrix}$ is a permutation of the vector $\begin{bmatrix}
        W_1\\\vdots\\W_k
    \end{bmatrix}$, and $[I_m|\,\bfP|\,\bfQ]$ is the corresponding permutation of columns of $[\bfV_1|\dots|\bfV_k]$ such that the above equation holds.
    \item $\bfP$ has dimension $m\times m$   and $\bfQ$ has dimension $m\times \bar m$ where $\bar m=\sum_{i\in[k]} m_i- 2m \geq 0$. 
    Moreover, $U$ and $V$ have dimension $m\times 1$, and $Z$ has dimension $\bar m \times 1$. 
\end{itemize} 
\begin{assumption}[Same-sender pairing]
\label{assumption:same_sender_pairing}
Suppose that the following condition holds.
\begin{equation}
\label{eq:Ui_Vi_same_sender}
    \text{For every } j\in[m], \text{ the symbols } U_j \text{ and } V_j
    \text{ belong to the same sender.}
\end{equation}
\end{assumption}    
We claim that there exists an $L=3$ instance coding scheme in which each Alice$_i$ sends $2m_i$ qudits. Equivalently, the normalized cost tuple is
\[
    \Delta
    =
    \left(
    \frac{2m_1}{3},\dots,\frac{2m_k}{3}
    \right).
\]

\begin{proof}
We first explain the idea of the construction. For three instances, Bob wants to
recover
\[
    \begin{bmatrix}
        U^{(1)}+\bfP V^{(1)}+\bfQ Z^{(1)}\\
        U^{(2)}+\bfP V^{(2)}+\bfQ Z^{(2)}\\
        U^{(3)}+\bfP V^{(3)}+\bfQ Z^{(3)}
    \end{bmatrix}.
\]
The goal is to realize this vector as $M_x\calX+M_z\calZ$  for some matrices $M_x,M_z$ satisfying the SSO condition $\Omega(M_x,M_z)=0$ and corresponding vectors $\calX$ and $\calZ$. Once this is done, the Sum Box protocol can compute $M_x\calX+M_z\calZ$ exactly.

The construction below splits the desired output into the following two complementary pieces.
\begin{align}
    M_x\calX
    =
    \begin{bmatrix}
        \bfP V^{(1)}+\bfQ Z^{(1)}\\
        U^{(2)}\\
        U^{(3)}
    \end{bmatrix}, \text{~and~}
    M_z\calZ
    =
    \begin{bmatrix}
        U^{(1)}\\
        \bfP V^{(2)}+\bfQ Z^{(2)}\\
        \bfP V^{(3)}+\bfQ Z^{(3)}
    \end{bmatrix}. \label{eq:restricted_transform_trick}
\end{align}
Thus, their sum is exactly the three desired outputs. The nontrivial point is
that $M_x$ and $M_z$ are chosen so that the SSO condition holds. This is the
reason for the special block $\bfP\bfP^T+\bfQ\bfQ^T$ in $M_x$ and the block
$\bfP^T$ in $M_z$ below: these terms cancel the cross terms in
$M_xM_z^T-M_zM_x^T$.

There is one further issue beyond the algebraic identity $M_x\mathcal X+M_z\mathcal Z=Y^{[3]}$. The vectors $\mathcal X$ and $\mathcal Z$ need to be specified such that each coordinate pair $(\mathcal X_i,\mathcal Z_i)$ must be known by the same sender. 
We verify that this holds after defining $\mathcal X$ and $\mathcal Z$.  

Now define 
\begin{equation}\label{eq:Mx_Mz_of_eg_of_restricted}
    \begin{split} 
    M_x =& \begin{bmatrix}
        \bfP & \bfQ & 0_{m\times m} & 0_{m\times m} & 0_{m\times \bar m}  & 0_{m\times  m}\\
        0_{m\times m} & 0_{m\times \bar m} & \bfP\bfP^T +\bfQ\bfQ^T & I_m & 0_{m\times \bar m} & 0_{m\times  m}\\
        0_{m\times m} & 0_{m\times \bar m} & 0_{m\times m} & 0_{m\times m} & 0_{m\times \bar m} & I_m
    \end{bmatrix}, \text{~and}\\
    M_z =&  \begin{bmatrix}
        0_{m\times m} & 0_{m\times \bar m}  & I_m & 0_{m\times m} & 0_{m\times \bar m}  & 0_{m\times   m}\\
        \bfP & \bfQ& 0_{m\times m} &0_{m\times m} & 0_{m\times \bar m} & \bfP^T\\
        0_{m\times m} & 0_{m\times \bar m} & 0_{m\times m} &  \bfP & \bfQ & 0_{m\times   m}
    \end{bmatrix}.
    \end{split}
\end{equation} 
The block structure above is designed precisely so that the SSO condition is satisfied, i.e., $\Omega(M_x,M_z)=0$
as verified in Appendix~\ref{eq:restricted_interval_eq:app}. Since $[M_x|M_z]$ has $3m$ rows, we have $3m\ge\text{rank}([M_x|M_z]).$ We note that 
\[\text{rank}([M_x|M_z]) \ge      \text{rank}\left(\begin{bmatrix}
        H &  &  \\ & H & \\ & & H
    \end{bmatrix}\right)\] 
    where $H = [I_m|\bfP|\bfQ]$, because it can be obtained by removing and permuting some columns of $[M_x|M_z]$. Since $[I_m|\bfP|\bfQ]$ is obtained by permuting columns of $[\bfV_1|\dots|\bfV_k]$, we conclude that $\text{rank}([M_x|M_z])  = 3 \cdot \text{rank}([\bfV_1|\dots|\bfV_k])=3m$ by Assumption \ref{assumption:lin_ind_assumption}.
 
Note that both $M_x$ and $M_z$ have dimension $3m\times (2\sum_{i\in[k]} m_i)$, and $2\sum_{i\in[k]} m_i =  4m+2\bar m$, and $3m\le 4m \le 4m+2\bar m$. Therefore, all three conditions of Lemma \ref{lemma:kappa,n-sum_box} are satisfied with identity LIT (\textit{cf. }Remark~\ref{remark:LIT_N_sum_box}). Then, a $(3m,  2\sum_{i\in[k]} m_i)$-Sum Box can compute $M_x \calX + M_z \calZ$ assuming that $\calX_i$ and $\calZ_i$ belong to the same sender for $i = 1, \dots, 2\sum_{i\in[k]}m_i$. Now, set 
\begin{equation}\label{eq:lemma7_eg_eq3}
\begin{split}
   &\calX^T = \begin{bmatrix}
       \left( V^{(1)} \right)^T&\left( Z^{(1)} \right)^T&\left( 0_{m\times 1}\right)^T&\left( U^{(2)} \right)^T&\left( 0_{\bar m \times 1} \right)^T&\left( U^{(3)}
\right)^T    \end{bmatrix}, \text{~and}\\ 
    &\calZ^T = 
    \begin{bmatrix}
       \left( V^{(2)} \right)^T&\left( Z^{(2)} \right)^T&\left( U^{(1)} \right)^T&\left( V^{(3)} \right)^T&\left( Z^{(3)} \right)^T&\left( 0_{m\times 1}\right)^T
    \end{bmatrix}.
\end{split} 
\end{equation}
It can be verified that \eqref{eq:restricted_transform_trick} holds.

It remains to assign the subsystems $\mathcal Q_1,\dots,\mathcal Q_{2\sum_{i\in[k]}m_i}$ to the senders. Recall that the $i$-th Sum Box subsystem is encoded using the coordinate pair $(\mathcal X_i,\mathcal Z_i)$. Hence this subsystem must be assigned to a sender
who knows both entries of this pair.

\begin{table}[!t]
\centering
\renewcommand{\arraystretch}{1.35}    
\setlength{\extrarowheight}{1.2pt}    
\begin{tabular}{|c|c|c|c|}\hline
Case & Range of $i$  &  Form of $(\calX_i,\calZ_i)$ & Assign $\calQ_i$ to sender who has \\ \hline
(1) &  $i\in \{1,\dots,m\}$ & $(V_j^{(1)},V_j^{(2)})$ for some $j\in[m]$  & $V_j$\\ \hline
(2) &  $i\in \{m+1,\dots,m + \bar m\}$ & $(Z_j^{(1)},Z_j^{(2)})$  for some $j\in[\bar m]$ & $Z_j$\\ \hline
(3) &  $i\in \{m + \bar m + 1,\dots,2m+\bar m\}$ & $(0,U_j^{(1)})$ for some $j\in[m]$  & $U_j$\\ \hline
(4) &  $i\in \{2m+\bar m + 1 ,\dots, 3m + \bar m\}$ 
& $(U_j^{(2)}, V_j^{(3)})$ for some $j\in[m]$  
& both $U_j$ and $V_j$\\ \hline
(5) &  $i\in \{3m+\bar m + 1 ,\dots, 3m + 2\bar m\}$ & $(0,Z_j^{(3)})$  for some $j\in[\bar m]$ & $Z_j$\\ \hline
(6) &  $i\in \{3m+2\bar m + 1 ,\dots, 4m + 2\bar m\}$ & $(U_j^{(3)},0)$ for some $j\in[m]$  & $U_j$\\ \hline
\end{tabular}
\vspace{3mm}
\caption{Assignment of Sum Box subsystems to senders in the base-case coding scheme. \label{table:eg7_assignment_of_subsystem}  } 
\end{table} 
From the definitions of $\mathcal X$ and $\mathcal Z$ in \eqref{eq:lemma7_eg_eq3}, the coordinate pairs fall into the six types corresponding to the rows of Table~\ref{table:eg7_assignment_of_subsystem}. Most types involve either two copies of the same variable block, such as $(V_j^{(1)},V_j^{(2)})$ or $(Z_j^{(1)},Z_j^{(2)})$, or one variable paired with zero. The only mixed type is $(U_j^{(2)},V_j^{(3)})$.
This pair can be encoded by a single sender precisely because Assumption~\ref{assumption:same_sender_pairing} guarantees that $U_j$ and $V_j$
belong to the same sender. 


We now calculate the number of qudits sent by each sender. Let $i\in[k]$, and $a,b$, and $c$ be the number of components of $W_i$ that go to $U,V,Z$ respectively. Note that $a+b+c=m_i$. By \eqref{eq:Ui_Vi_same_sender} we have 
$a=b$. For $\ell \in\{1,\dots,6\}$, let $d_\ell$ be the number of subsystems assigned to Alice$_i$ due to Case $(\ell)$ in Table \ref{table:eg7_assignment_of_subsystem}. For $j\in [4m+2 \bar m],$   we assign $\calQ_j$ to Alice$_i$ based on whether she has $V_k$ or $U_k$ or $Z_k$ for some $k$. The following  equations hold.
\begin{align*}
    d_1 = b, d_2 = c, d_3 = a, d_4 = a= b, d_5 = c, d_6 = a. 
\end{align*}
Therefore, the total number of qudits sent by Alice$_i$ is
\begin{align*}
    d_1+\cdots+d_6 = b+c+a+a+c+a = 4a+2c  = 2(a+b+c) = 2m_i,
\end{align*}
where we used $a=b$  and  $a+b+c=m_i$. Since $L=3$, the cost tuple is $\Delta = \big( \frac{2m_1}{3},\dots,\frac{2m_k}{3} \big )$.
\end{proof}

\end{example} 
\begin{proof}[Proof of Theorem \ref{theorem:first_main_result_informal}(b)]
Since $(\bbF_q,k,\bfV_1,\dots,\bfV_k)$ is a Restricted Interval-$[0,1]$ LC,  there exists an integer $L\in\bbN$ and integers $a_i\in\mathbb Z_{\ge 0}$ with $a_i\le \frac{m_iL}{2},\,i\in[k],$ such that the following holds. For each $i\in[k]$, one can choose $a_i$ columns from $\bfV_i^{[L]}$, denoted by $A_i$, such that $[A_1|A_2|\cdots|A_k]$ is an invertible $mL\times mL$ matrix. 
For each $i\in[k]$, there exists  a column permutation matrix $\bar P_i$ such that  
\begin{equation}\label{eq:column_permutation}
    \bfV_i^{[L]} \bar P_i = [A_i | B_i |C_i],
\end{equation}
where $B_i$ is a block with the same dimension as $A_i$, i.e. its dimension is $m\cdot L\times a_i$, and $C_i$ contains the rest of the columns; $C_i$ has dimension $m\cdot L\times (m_i\cdot L-2a_i)$. Since $a_i\le \frac{m_i\cdot L}{2}$ holds, $C_i$ is well defined and $C_i$ being empty is allowed. 

Denote $A=[A_1|A_2|\dots|A_k]$, $B=[B_1|\dots|B_k]$, $C=[C_1|\dots|C_k]$. Note $B$ has dimension $m\cdot L \times m\cdot L$ as $\sum_{i\in[k]} a_i = m\cdot L$, and $C$ has dimension $m\cdot L\times (\sum_{i\in[k]} m_i\cdot L -2m\cdot L)$.  We define $\left(    (\bar P_i)^{-1} W_i^{[L]} \right)^T = \begin{bmatrix}
        U_i&V_i& Z_i
    \end{bmatrix}$ 
where $U_i,V_i$, and $Z_i$ have dimension $a_i \times 1, a_i\times 1, (m_i\cdot L-2a_i)\times 1$ respectively, and define $U,V,Z$ by $U^T = \begin{bmatrix}
    U_1^T&\dots&U_k^T
\end{bmatrix}$, $V^T = \begin{bmatrix}
    V_1^T&\dots&V_k^T
\end{bmatrix}$, $Z^T = \begin{bmatrix}
    Z_1^T&\dots&Z_k^T
\end{bmatrix}.$ Then, we have
\begin{align*}
    Y^{[L]} =&\sum_{i\in[k]} \bfV_i^{[L]} W_i^{[L]}   = \sum_{i\in [k]}\big  (\bfV_i^{[L]} \bar P_i\big ) \cdot \big ( (\bar P_i)^{-1}W_i^{[L]}\big ) =    \sum_{i\in [k]} [A_i|B_i|C_i] \begin{bmatrix}
        U_i\\V_i\\ Z_i
    \end{bmatrix}\\
    =& \sum_{i\in [k]}  A_i U_i + B_i V_i + C_i Z_i = [A_1|\dots |A_k]\begin{bmatrix}
        U_1\\\vdots\\U_k
    \end{bmatrix} + [B_1|\dots |B_k]\begin{bmatrix}
        V_1\\\vdots\\V_k
    \end{bmatrix}  + [C_1|\dots |C_k]\begin{bmatrix}
        Z_1\\\vdots\\Z_k
    \end{bmatrix}\\
    =& AU+BV+CZ.
\end{align*}
Recall that $A$ is invertible. Consider 
\begin{align*}
   A^{-1}  Y^{[L]} = U+(A^{-1}B)V+(A^{-1}C)Z.
\end{align*}
Then, we observe that
\begin{itemize}
    \item $A^{-1}B$ has dimension $m\cdot L\times m\cdot L$, $A^{-1}C$ has dimension $m\cdot L\times (\sum_{i\in[k]} m_i\cdot L -2m\cdot L)$
    \item $U,V$ have dimension $m\cdot L\times 1,$ $Z$ has dimension $(\sum_{i\in[k]} m_i\cdot L -2m\cdot L)\times 1$.
\item The vectors $U$ and $V$ satisfy the same-sender pairing required in
Assumption~\ref{assumption:same_sender_pairing}. To see this, fix a sender
Alice$_i$. By construction,
\[
    (\bar P_i)^{-1}W_i^{[L]}
    =
    \begin{bmatrix}
        U_i\\V_i\\Z_i
    \end{bmatrix},
    \qquad
    U_i,V_i\in\bbF_q^{a_i}.
\]
Thus the $r$-th coordinate of $U_i$ and the $r$-th coordinate of $V_i$ both come
from Alice$_i$'s own data vector, for every $r\in[a_i]$. After concatenating
these blocks over all senders, each coordinate pair $(U_j,V_j)$ belongs to a
single sender. Hence Assumption~\ref{assumption:same_sender_pairing} in
Example~\ref{example:lemma_Restricted Interval-$[0,1]$_coding_scheme} is
satisfied.
\end{itemize}

Therefore, $A^{-1}  Y^{[L]} = U+(A^{-1}B)V+(A^{-1}C)Z$ reduces to the setting of Example \ref{example:lemma_Restricted Interval-$[0,1]$_coding_scheme}, which allows us to compute   $A^{-1}  Y^{[L]} = U+(A^{-1}B)V+(A^{-1}C)Z$. Once this is computed, Bob can recover  $Y^{[L]} =AU+BV+CZ$ by left-multiplying by $A$. 
When we use Example \ref{example:lemma_Restricted Interval-$[0,1]$_coding_scheme}, the coding scheme computes three instances of  $A^{-1}  Y^{[L]} = U+(A^{-1}B)V+(A^{-1}C)Z$. The total number of qudits sent by Alice$_i$ is $2 m_i\cdot L$, and the cost tuple, which is normalized by the number of instances, is $\Delta = \big (\frac{2m_1\cdot L}{3\cdot L},\dots, \frac{2m_k\cdot L}{3\cdot L}\big ) =\big (\frac{2m_1}{3},\dots, \frac{2m_k}{3}\big )$. 
\end{proof}

\begin{remark}
The reduction in the proof above can be seen concretely in
Example~\ref{example:restricted_interval_example}. In that example, $L=1$ and
the selected columns give
\[
    A=[A_1\mid A_2]
    =
    \begin{bmatrix}
        1&0\\
        0&1
    \end{bmatrix}.
\]
After permuting columns within each sender's block in the same way as in
\eqref{eq:column_permutation}, we can write
\begin{align*}
    Y
    &=
    {\begin{bmatrix}
        1&1\\
        0&1
    \end{bmatrix}}
    \begin{bmatrix}
        W_{11}\\
        W_{12}
    \end{bmatrix}
    +
    {\begin{bmatrix}
        1&0\\
        1&1
    \end{bmatrix}}
    \begin{bmatrix}
        W_{21}\\
        W_{22}
    \end{bmatrix}
    +
    {\begin{bmatrix}
        0&1\\
        1&0
    \end{bmatrix}}
    \begin{bmatrix}
        W_{31}\\
        W_{32}
    \end{bmatrix} \\
    &=
    \begin{bmatrix}
        1&0\\
        0&1
    \end{bmatrix}
    \underbrace{\begin{bmatrix}
        W_{11}\\
        W_{22}
    \end{bmatrix}}_{U}
    +
    \underbrace{\begin{bmatrix}
        1&1\\
        1&1
    \end{bmatrix}}_{\bfP}
    \underbrace{\begin{bmatrix}
        W_{12}\\
        W_{21}
    \end{bmatrix}}_{V}
    +
    \underbrace{\begin{bmatrix}
        0&1\\
        1&0
    \end{bmatrix}}_{\bfQ}
    \underbrace{\begin{bmatrix}
        W_{31}\\
        W_{32}
    \end{bmatrix}}_{Z} \\
    &= U+\bfP V+\bfQ Z.
\end{align*}
Since $A=I_2$ in this example, this is exactly the transformed problem
$A^{-1}Y^{[L]}=U+(A^{-1}B)V+(A^{-1}C)Z$ appearing in the proof. Applying
Example~\ref{example:lemma_Restricted Interval-$[0,1]$_coding_scheme} to this
decomposition computes three copies of this transformed problem, and Bob then
recovers the original outputs by multiplying by $A$.

The important structural point is that the selected $U$-symbols can be paired
with $V$-symbols from the same senders: $U_1$ and $V_1$ both belong to
Alice$_1$, while $U_2$ and $V_2$ both belong to Alice$_2$. This same-sender
pairing is exactly the property used by the base-case Sum Box protocol.
    
\end{remark}

\begin{remark}[Role of the Restricted Interval condition]
The inequality $a_i\le \frac{m_iL}{2}$ in Definition~\ref{def:Restricted Interval-$[0,1]$} is used to ensure that, after selecting the $a_i$ columns $A_i$ from $\bfV_i^{[L]}$, there are still at least $a_i$ remaining columns from the same sender. These remaining columns form the block $B_i$ in \eqref{eq:column_permutation}. Consequently, Alice$_i$'s data can be locally partitioned as
\[
    (\bar P_i)^{-1}W_i^{[L]}
    =
    \begin{bmatrix}
        U_i\\
        V_i\\
        Z_i
    \end{bmatrix},
    \qquad
    U_i,V_i\in\bbF_q^{a_i}.
\]
Thus, for every coordinate $r\in[a_i]$, the symbols $(U_i)_r$ and $(V_i)_r$ are both held by Alice$_i$. After concatenating the sender blocks, this gives the same-sender pairing (\textit{cf.} Assumption~\ref{assumption:same_sender_pairing}).

This pairing is essential in the base-case protocol of Example~\ref{example:lemma_Restricted Interval-$[0,1]$_coding_scheme}. In Case~(4) of Table~\ref{table:eg7_assignment_of_subsystem}, one Sum Box subsystem must encode the pair $(U_j^{(2)},V_j^{(3)})$. If $U_j$ and $V_j$ belonged to different senders, then no single sender would know both entries of this pair, and the local encoding required by the Sum Box protocol would not be valid.

Example~\ref{example:interval_not_restricted} illustrates what can go wrong without the Restricted Interval-$[0,1]$ condition. In that example, for every $L\ge 1$, any invertible $2L\times 2L$ submatrix must include all $L$ columns of $\bfV_2^{[L]}$. Hence, all $L$ symbols of Alice$_2$ are forced into the $U$-block. Since Alice$_2$ has only $L$ symbols in total, there are no remaining symbols of Alice$_2$ available to form a matching $V$-block of the same size. Consequently, after the $U,V,Z$ decomposition, one cannot guarantee that every coordinate $U_j$ can be paired with a coordinate $V_j$ held by the same sender. This is precisely the obstruction that the Restricted Interval-$[0,1]$ condition rules out. 

The Restricted Interval-$[0,1]$ condition is therefore a convenient sufficient condition for producing the same-sender pairing needed by the base-case construction. However, it is not necessary for every possible coding scheme: the base-case construction applies whenever one can find a decomposition
\[
    Y=U+\bfP V+\bfQ Z
\]
satisfying the same-sender pairing condition, even if that decomposition is obtained by another argument. In the upcoming Section~\ref{subsec:direct_sum_coding_scheme}, we use related decompositions for LC problems that lie
outside the Restricted Interval-$[0,1]$ class.
\end{remark}

\subsection{On  Direct-Sum  Problems}
\label{subsec:direct_sum_coding_scheme}  
We now discuss a class of direct-sum examples that can be obtained by generalizing Example \ref{example:subadditivity2} in Section \ref{subsec:main_contribution_direct-sum}.

We first introduce a structural condition on a pair $(\mathrm{LC}_1,\mathrm{LC}_2)$ of LC problems. This condition abstracts the algebraic cancellation pattern appearing in Example~\ref{example:subadditivity2} and defines a family of pairs for which the same joint coding idea applies. 



\begin{condition}\label{condition:direct_sum_condition}
 We denote  $\text{LC}_1=(\bbF_q, k, \bfV_1,\dots,\bfV_k)$ and  $\text{LC}_2=(\bbF_q, k, \bfV_1',\dots,\bfV_k')$. For $i\in[k]$, let $\bfV_i$ have dimension $m\times m_i$ and $\bfV_i'$ have dimension $m\times m_i'$. Then, the pair $(\text{LC}_1,\text{LC}_2)$ is said to satisfy Condition \ref{condition:direct_sum_condition} if the following holds.

    \begin{enumerate}[label = (\alph*)] 
     \item Both $\text{LC}_1$ and $\text{LC}_2$ satisfy Assumption \ref{assumption:lin_ind_assumption}.
    \item $[\bfV_1|\dots|\bfV_k]$ has dimension $m\times 2m$, and $[\bfV_1'|\dots|\bfV_k']$ has dimension $m\times m$, i.e., we have $\sum_{i\in[k]}m_i=2m$, and $\sum_{i\in[k]}m_i'=m$.

    \item For each $i\in[k]$, we can express $\bfV_i = [A_i|B_i]$ such that (a) $[B_1|\dots|B_k]$ is invertible and (b) $A_i$ has the same dimension as $\bfV_i'$ for every $i\in[k]$, i.e. both $A_i$ and $\bfV_i'$ have dimension $m\times m_i'$. 
    \end{enumerate}
\end{condition}

\begin{theorem}\label{theorem:direct_sum_coding_formal} 
Denote   $\text{LC}_1=(\bbF_q, k, \bfV_1,\dots,\bfV_k)$ and  $\text{LC}_2=(\bbF_q, k, \bfV_1',\dots,\bfV_k')$. Suppose $(\text{LC}_1,\text{LC}_2)$ satisfies Condition \ref{condition:direct_sum_condition}. Then the optimal total cost $\Gamma^{*,\mathsf Q}_{\mathrm{LC}_1\oplus\mathrm{LC}_2}=2m.$ (recall that the superscript $\mathsf Q$ denotes the noiseless
quantum many-to-one model).
\end{theorem}

\begin{proof}
\label{example:base_case_direct_sum} Let the LC problems $\text{LC}_1=(\bbF_q,k,\bfV_1,\dots,\bfV_k)$ and  $\text{LC}_2=(\bbF_q,k,\bfV_1',\dots,\bfV_k')$ be associated with  data vectors $\{W_i\}_{i\in[k]}$ and $\{W_i'\}_{i\in[k]}$. 

Let $\hat m_i$ denote the number of columns of $B_i$, i.e.  $\hat m_i = m_i - m_i'$, and write   $W_i = \begin{bmatrix}
    U_i\\T_i
\end{bmatrix}$ where $U_i\in\bbF_q^{m_i'}$ and $T_i\in\bbF_q^{\hat m_i}$. Define 
\[
    U^T=\begin{bmatrix}U_1^T&\cdots&U_k^T\end{bmatrix},\quad
    T^T=\begin{bmatrix}T_1^T&\cdots&T_k^T\end{bmatrix},\quad
    (W')^T=\begin{bmatrix}(W_1')^T&\cdots&(W_k')^T\end{bmatrix}.
\]
Furthermore, $A=[A_1|\dots|A_k],B=[B_1|\dots|B_k]$, and $\bfV'=[\bfV_1'|\dots|\bfV_k']$. As the dimensions of $[\bfV_1|\dots|\bfV_k]$ and $[\bfV_1'|\dots|\bfV_k']$ are $m\times 2m$ and $m\times m$, respectively, we have $\sum_{i\in[k]}m_i' = m$ and $\sum_{i\in[k]}\hat m_i  = 2m-m = m$. This implies that  dimensions of each $U,T$, and $W'$ are $m\times 1$, and both $A$ and $B$ are of dimension $m\times m$. Then, 
\begin{equation}
    \begin{split}\label{eq:example_base_case_direct_sum}
            & \bfV_1 W_1 + \dots + \bfV_k W_k\\
    &= [A_1|B_1] \begin{bmatrix}
        U_1\\T_1
    \end{bmatrix} + \dots + [A_k|B_k] \begin{bmatrix}
        U_k\\T_k
    \end{bmatrix}\\
    &= [A_1|\dots|A_k] \begin{bmatrix}
    U_1\\\vdots\\U_k
\end{bmatrix}+ [B_1|\dots|B_k]\begin{bmatrix}
    T_1\\\vdots\\T_k
\end{bmatrix}=AU+BT.
    \end{split}
\end{equation}

Since $A_i$ and $\bfV_i'$ have the same number of columns, $U_i$ and $W_i'$ have the same dimension for each $i\in[k]$. Moreover, both $U_i$ and $W_i'$ are held by Alice$_i$. Since
\[
    U^T=\begin{bmatrix}U_1^T&\cdots&U_k^T\end{bmatrix},
    \qquad
    (W')^T=\begin{bmatrix}(W_1')^T&\cdots&(W_k')^T\end{bmatrix},
\]
we therefore have the following same-sender pairing:
\begin{equation}
\label{eq:direct_sum_same_sender_pairing}
    \text{For every } j\in[m], \text{ the symbols } U_j \text{ and } W_j'
    \text{ belong to the same sender.}
\end{equation}
Thus, $\text{LC}_1\oplus \text{LC}_2$ can be written
\begin{equation}\label{eq:example_base_case_direct_sum_2}
    \begin{split}
    &\begin{bmatrix}
        \bfV_1&\\
        &\bfV_1'
    \end{bmatrix} \begin{bmatrix}
        W_1\\W_1'
    \end{bmatrix} +\dots +\begin{bmatrix}
        \bfV_k&\\
        &\bfV_k'
    \end{bmatrix} \begin{bmatrix}
        W_k\\W_k'
    \end{bmatrix}\\
    &= \begin{bmatrix}
        \bfV_1W_1+\dots+\bfV_kW_k\\
        \bfV_1'W_1'+\dots+\bfV_k'W_k'
    \end{bmatrix}
    \overset{\text{By \eqref{eq:example_base_case_direct_sum}}}{=}
    \begin{bmatrix}
        AU+BT\\
        \bfV'W'
    \end{bmatrix}.
    \end{split}
\end{equation} 
We claim that there is a coding scheme computing $\text{LC}_1\oplus\text{LC}_2$ that works with a single ($L=1$) instance and has total cost $\Gamma=2m$. Towards this end, we make the following observation.
\begin{equation}
    \begin{split}\label{eq:direction_sum_eq_29}
        &\Omega\left(\begin{bmatrix}
            A&B\\
            0_{m\times m}&0_{m\times m}
        \end{bmatrix}, \begin{bmatrix}
            0_{m\times m}&0_{m\times m}\\
            \bfV' & -\big( B^{-1} A(\bfV')^T \big)^T
        \end{bmatrix}\right)\\
        &\overset{(a)}{=}\Omega\left(\begin{bmatrix}
            A\\
            0_{m\times m}
        \end{bmatrix}, \begin{bmatrix}
            0_{m\times m}\\
            \bfV'
        \end{bmatrix}\right)+
        \Omega\left(\begin{bmatrix}
            B\\
            0_{m\times m}
        \end{bmatrix}, \begin{bmatrix}
           0_{m\times m}\\
            -\big( B^{-1} A(\bfV')^T \big)^T
        \end{bmatrix}\right)\\
        &=\begin{bmatrix}
            0 & A(\bfV')^T\\
            -\bfV'A^T & 0
        \end{bmatrix}-\begin{bmatrix}
            0 & A(\bfV')^T\\
            -\bfV'A^T & 0
        \end{bmatrix}=0,
    \end{split}
\end{equation}
where $(a)$ follows from Proposition \ref{prop:SSO_proposition}\ref{prop:SSO_partition}. Set 
\begin{align*}
   M_x = \begin{bmatrix}
            A&B\\
            0_{m\times m}&0_{m\times m}
        \end{bmatrix}, \text{~and~} M_z = \begin{bmatrix}
            0_{m\times m}&0_{m\times m}\\
            \bfV' & -\big( B^{-1} A(\bfV')^T \big)^T
        \end{bmatrix}.
\end{align*}
Then, \eqref{eq:direction_sum_eq_29} implies that $\Omega(M_x,M_z) = 0$. Next, $[M_x|M_z]$ has rank $2m$, because it contains the submatrix $\begin{bmatrix}
    B&0_{m\times m}\\
    0_{m\times m}&\bfV'
\end{bmatrix}$, and both $B$ and $\bfV'$ are invertible by our assumption. Note that the dimensions of both $M_x,M_z$ are $2m\times 2m$; let $\kappa=N=2m$. 
Thus, all three conditions of  Lemma \ref{lemma:kappa,n-sum_box} are satisfied with identity LIT (\textit{cf. }Remark~\ref{remark:LIT_N_sum_box}). This implies that there exists a $(2m,2m)$-Sum Box computing $M_xX+M_zZ$. Set $X = \begin{bmatrix}
            U\\T
        \end{bmatrix}, Z = \begin{bmatrix}
            W'\\0_{m\times 1}
        \end{bmatrix}.$ Then, we have 
\begin{align*}
    &M_xX+M_zZ \\
    =&  \begin{bmatrix}
            A&B\\
            0_{m\times m}&0_{m\times m}
        \end{bmatrix} \begin{bmatrix}
            U\\T
        \end{bmatrix} +  \begin{bmatrix}
            0_{m\times m}&0_{m\times m}\\
            \bfV' & -\big( B^{-1} A(\bfV')^T \big)^T
        \end{bmatrix} \begin{bmatrix}
            W'\\0_{m\times 1}
        \end{bmatrix}\\
        =&\begin{bmatrix}
            AU+BT\\
            \bfV' W'
        \end{bmatrix}\\
        \overset{\text{By }\eqref{eq:example_base_case_direct_sum_2}}{=} & \begin{bmatrix}
        \bfV_1&\\
        &\bfV_1'
    \end{bmatrix} \begin{bmatrix}
        W_1\\W_1'
    \end{bmatrix} +\dots+ \begin{bmatrix}
        \bfV_k&\\
        &\bfV_k'
    \end{bmatrix} \begin{bmatrix}
        W_k\\W_k'
    \end{bmatrix}.
\end{align*}
Let $\calQ=\calQ_1\dots\calQ_{2m}$ be the quantum system of the $(2m,2m)$-Sum Box computing $M_xX+M_zZ$.

If $i\in\{1,\dots,m\}$, then $(X_i,Z_i)$ is of the form $(U_j,W_j')$ for some coordinate $j$. By \eqref{eq:direct_sum_same_sender_pairing}, both entries are held by the same Alice, so we assign $\calQ_i$ to that Alice.

If $i\in\{m+1,\dots,2m\}$, then $(X_i,Z_i)$ is of the form $(T_{i-m},0)$, so we assign $\calQ_i$ to the sender who holds $T_{i-m}$. Thus the protocol transmits exactly $2m$ qudits in total. 

For the converse, let $\Delta$ be the cost tuple of any zero-error coding scheme for $\mathrm{LC}_1\oplus\mathrm{LC}_2$. By Lemma~\ref{lemma:quantum_many_to_one_lower_bound},
\[
    \Gamma(\Delta)
    \ge
    \text{rank}
    \left(
    \left[
        \begin{bmatrix}\bfV_1&0\\0&\bfV_1'\end{bmatrix}
        \middle|
        \cdots
        \middle|
        \begin{bmatrix}\bfV_k&0\\0&\bfV_k'\end{bmatrix}
    \right]
    \right).
\]
The matrix inside the rank is block diagonal after grouping columns:
\[
    \begin{bmatrix}
        \bfV_1&\cdots&\bfV_k&0&\cdots&0\\
        0&\cdots&0&\bfV_1'&\cdots&\bfV_k'
    \end{bmatrix}.
\]
Therefore,
\[
    \Gamma(\Delta)
    \ge
    \text{rank}([\bfV_1|\cdots|\bfV_k])
    +
    \text{rank}([\bfV_1'|\cdots|\bfV_k'])
    =
    2m,
\]
where the last equality follows from Condition~\ref{condition:direct_sum_condition}. Thus $\Gamma_{\mathsf Q}^*(\mathrm{LC}_1\oplus\mathrm{LC}_2)\ge 2m$. The achievability above gives the reverse inequality. 
\end{proof}

\begin{remark}
    The construction from Theorem \ref{theorem:direct_sum_coding_formal} exploits the same-sender-pairing ideas as the construction in  Theorem \ref{theorem:first_main_result_informal}(b). Note that here, we pick $\text{LC}_1$ and $\text{LC}_2$ that satisfy Condition \ref{condition:direct_sum_condition}; this allows us to design a specific precoding scheme that directly applies to it rather than having to solve the minimization in Problem \ref{problem:precoding}. Moreover, we are able to show that the cost is optimal based on the lower bound from Lemma \ref{lemma:quantum_many_to_one_lower_bound}.
\end{remark}
 

\begin{remark}
    It can be seen that there are instances (of $\text{LC}_1$) where the optimal cost for the $\text{LC}_1$ problem alone is strictly larger than $m$. For instance, Example \ref{example:subadditivity}  demonstrates an instance $\text{LC}_1$ where the cost is strictly larger than $m=2$. The subsequent discussion in Examples \ref{example:subadditivity1} and \ref{example:subadditivity2} demonstrates that the cost of $\text{LC}_1 \oplus \text{LC}_2$ is strictly lower than the sum of the optimal costs of $\text{LC}_1$ and $\text{LC}_2$.
\end{remark}
In contrast, the optimal total cost under the classical many-to-one model is additive.
\begin{proposition}\label{prop:classical_many_to_one_additive}  
Suppose $\text{LC}_1=(\bbF_q,k,\bfV_1,\dots,\bfV_k)$ and  $\text{LC}_2=(\bbF_q,k,\bfV_1',\dots,\bfV_k')$ are two LC problems over the same field and with the same number of senders, and suppose both satisfy Assumption~\ref{assumption:lin_ind_assumption}.  Denote by $\Gamma^{*,\mathsf{C}}_{\text{LC}_1}$, $\Gamma^{*,\mathsf{C}}_{\text{LC}_2}$, $ \Gamma^{*,\mathsf{C}}_{\text{LC}_1 \oplus \text{LC}_2}$ their optimal total costs under the classical model. Then
  \begin{align*}
      \Gamma^{*,\mathsf{C}}_{\text{LC}_1}+\Gamma^{*,\mathsf{C}}_{\text{LC}_2}= \Gamma^{*,\mathsf{C}}_{\text{LC}_1 \oplus \text{LC}_2}
  \end{align*}
\end{proposition}
\begin{proof}
Applying  Proposition \ref{prop:classical_many_to_one_baseline} to $\text{LC}_1,\text{LC}_2,\text{LC}_1\oplus \text{LC}_2$, we have that 
\begin{align*}
    \Gamma_{\text{LC}_1}^{*,\mathsf{C}} = &\sum_{i\in[k]} \text{rank}(\bfV_i),\qquad
    \Gamma_{\text{LC}_2}^{*,\mathsf{C}} = \sum_{i\in[k]} \text{rank}(\bfV_i'), \text{~and}\\
    \Gamma_{\text{LC}_1\oplus \text{LC}_2}^{*,\mathsf{C}} = &\sum_{i\in[k]} \text{rank}\left(\begin{bmatrix}
        \bfV_i&\\&\bfV_i'
    \end{bmatrix}\right) = \sum_{i\in[k]} \text{rank}(\bfV_i)+ \text{rank}(\bfV_i')  =  \Gamma_{\text{LC}_1}^{*,\mathsf{C}}+ \Gamma_{\text{LC}_2}^{*,\mathsf{C}}.
\end{align*}
\end{proof} 

\section{Comparison with prior work} \label{sec:comparison_with_prior_works}

The works most closely related to ours are \cite{N_Sum_Box_JafarYao25,Summation_JafarYao25, Inverted_3_Sum_Box_JafarYao25,HuUlukus25}.  We compare our results with these works from three perspectives: the class of linear functions covered, the communication cost of the resulting schemes, and the computational complexity of finding the required precoding.

\paragraph{Relation to the $N$-Sum Box and its variants} The works in \cite{N_Sum_Box_JafarYao25} and \cite{Summation_JafarYao25} were among the first to consider linear combination problems in the distributed superdense-coding setting. In the $N$-Sum Box work \cite{N_Sum_Box_JafarYao25}, Bob wants to compute $y = M_x \bfx + M_z \bfz$ where $M_x,M_z$ are   $\kappa \times N$ matrices over $\bbF_q$. There are $N$ senders, and Alice$_i$ has $x_i$ and $z_i$. The $N$-Sum Box protocol
applies when
\[
    \Omega(M_x,M_z)=M_xM_z^T-M_zM_x^T=0,
    \qquad
    \text{rank}([M_x|M_z])=\kappa,
    \qquad
    \kappa\le N.
\] 
The same work also gives two extensions that are important for comparison: the locally invertible transform (LIT) characterization for $N$-output transfer matrices, and a reduction from $\kappa$ output coordinates to an $N$-output Sum Box when $\kappa\le N$. In our terminology, the $N$-Sum Box is a linear-computation protocol over a noiseless quantum many-to-one network.  

Our schemes use this protocol as a black box, but address LC problems for which the SSO condition may not hold directly. The main idea is to perform distributed precoding and, when useful, encoding across multiple instances, so that the transformed computation satisfies the hypotheses of the $N$-Sum Box.
 


In \cite{Summation_JafarYao25}, the setting involves computing a sum, i.e., Bob wants to compute \[y=x_1+\dots + x_N.\]  This is an important but   more restrictive  class of functions than the general LC problem considered here. The work of \cite{Inverted_3_Sum_Box_JafarYao25} gives a complete (information-theoretically optimal) characterization for the three-sender case. 
However, its derivation relies on a rather fine-grained bookkeeping of the algebraic relations among the three subspaces $\mathrm{col}(\bfV_1),\,\mathrm{col}(\bfV_2),\,\mathrm{col}(\bfV_3).$  Expressing the rate region involves considering the rank of the union of all possible submatrices of $\bfV_i, i = 1, \dots, 3$. 
%
This is still manageable for three senders, but the approach is inherently combinatorial: a direct extension to general $k$ would require tracking an exponentially growing family of subspace relations (essentially over all subsets of $[k]$), which would make the exact region computation and the search for an optimal coding strategy quickly intractable. In contrast, our approach directly targets arbitrary $k$ with relatively simple protocols that avoid enumerating the entire subspace lattice. Although the protocols may not always be optimal, they give closed-form achievable costs and apply to broad classes of LC problems. Table~\ref{tab:comparison_prior_work} summarizes the main differences between the closest prior construction and the schemes developed in this paper.

\begin{table*}[t!]
\caption{Comparison of the closest prior construction and the schemes in this paper.}
\label{tab:comparison_prior_work}
\centering
\scriptsize
\setlength{\tabcolsep}{3pt}
\renewcommand{\arraystretch}{1.15}
\begin{tabularx}{\textwidth}{@{}p{0.17\textwidth}p{0.29\textwidth}p{0.21\textwidth}X@{}}
\toprule
\textbf{Construction} 
& \textbf{Applicability} 
& \textbf{Cost} 
& \textbf{Main point} \\
\midrule
$(\kappa, N)$-Sum Box~\cite{N_Sum_Box_JafarYao25}
& SSO condition (may also hold after an LIT), \(\text{rank}([M_x|M_z])=\kappa\).
& \(N\) qudits.
& Primitive used to implement self-orthogonal linear computations. \\

Hu {\it et al.}~\cite{HuUlukus25}
& General LC problems under their full-rank assumptions.
& \(\bigl(\sum_i m_i+\mathbf{Prob}_2(\mathrm{LC})\bigr)/2\).
& Closest prior achievability scheme; compared below with Scheme I on the Interval-\([0,1]\) subclass. \\

This paper: Scheme I
& Interval-\([0,1]\) LC problems.
& \(\bigl(\sum_i m_i+\mathbf{Prob}_1(\mathrm{LC})\bigr)/2\).
& On this subclass, no larger cost than Hu {\it et al.} since
\(\mathbf{Prob}_1(\mathrm{LC})\le\mathbf{Prob}_2(\mathrm{LC})\);
the improvement can be strict. \\

This paper: Scheme II
& Restricted Interval-\([0,1]\) / Double-Basis-type instances.
& \(2(\sum_i m_i)/3\).
& Closed-form cost; can strictly improve  over Hu {\it et al.} in certain cases. \\

This paper: Direct sums
& Pairs satisfying Condition~\ref{condition:direct_sum_condition}.
& \(2m\), optimal for that class.
& Demonstrates strict subadditivity of the quantum many-to-one communication cost. \\

\bottomrule
\end{tabularx}  
\end{table*}



 
\paragraph{Comparison with Hu {\it et al.}} 
The work most closely related to ours is \cite{HuUlukus25}. It studies essentially the same noiseless quantum many-to-one model for LC with multiple instances, and proposes an achievable scheme based on enlarging the linear combination with auxiliary entangled qudits and optimizing precoding matrices.


In their formulation, the achieved total cost \(\Gamma\) has the form $\frac{\sum_i m_i+c}{2},$ where \(c\) is the optimum value of the following precoding problem.
\begin{problem}\label{problem:ulukus}
(Precoding Problem of~\cite{HuUlukus25}) Let $\text{LC}=(\bbF_q,k,\bfV_1,\dots,\bfV_k)$. Define 
$$\textbf{Prob}_2(\text{LC}) = \min_{\bfP} \text{rank}(\sum_{i\in [k]}\bfV_i P_i \bfV_i^T),$$
where the precoding  $\bfP=(P_1,\dots,P_k)$ is such that $P_i$ ranges over   invertible $m_i\times m_i$ matrices over $\bbF_q$.
\end{problem}

The total cost of our first coding scheme (\textit{cf.} Theorem \ref{theorem:first_main_result_informal}(a)) has the same structural form
\(
(\sum_{i} m_i + s)/2
\),
with $s=\textbf{Prob}_1(\text{LC})$ defined through an optimal precoding problem (\textit{cf.} Problem \ref{problem:precoding}) that is similar to (but different from) $\textbf{Prob}_2(\text{LC})$.   We show in Lemma~\ref{lemma:prob1<=prob2} (see below)  that there exists a mapping from the feasible set of \cite{HuUlukus25} into our feasible set that preserves feasibility and the objective value.   

\begin{lemma}\label{lemma:prob1<=prob2}
        Let $\text{LC}=(\bbF_q,k,\bfV_1,\dots,\bfV_k)$ be an Interval-$[0,1]$ problem. Then, $\textbf{Prob}_1(\text{LC})\le \textbf{Prob}_2(\text{LC})$.
\end{lemma}
\begin{proof}
Let $\bfP=(P_i)_{i=1}^k$ be  such that $\textbf{Prob}_2(\text{LC}) = \text{rank}(\sum_{i\in [k]}\bfV_i P_i \bfV_i^T)$. Then, we note that  
\begin{align*}
  &\frac{1}{2}\text{rank}\left(\sum_{i\in[k]}\Omega\left( 
        \begin{bmatrix}
            \bfV_iP_i\\0
        \end{bmatrix},\begin{bmatrix}
            0\\ \bfV_i
        \end{bmatrix}
     \right)\right) = \frac{1}{2} \text{rank}\left(\begin{bmatrix}
        0&\sum_{i\in [k]}\bfV_i P_i \bfV_i^T\\
        -\big(\sum_{i\in [k]}\bfV_i P_i \bfV_i^T\big)^T & 0
    \end{bmatrix}\right)\\
     & = \text{rank}\big(\sum_{i\in [k]}\bfV_i P_i \bfV_i^T\big) = \textbf{Prob}_2(\text{LC}).
\end{align*}
By Proposition \ref{prop:trivial_precoding_bound}, the conclusion follows.
\end{proof}

The inequality in Lemma~\ref{lemma:prob1<=prob2} shows that our first coding scheme is never worse than the scheme of \cite{HuUlukus25} on Interval-$[0,1]$ LC problems. The forthcoming Example~\ref{example:Ulukus_eg1} shows that this improvement can be strict. Furthermore, Example~\ref{example:Ulukus_eg2} below shows that, even for Restricted Interval-$[0,1]$ problems, the cost achieved by Theorem~\ref{theorem:first_main_result_informal}(b) can be strictly lower than that of \cite{HuUlukus25}. 

Our work also shows the subadditivity of $\Gamma^{*,\mathsf{Q}}$ under the noiseless quantum many-to-one model (\textit{cf.} Theorem \ref{theorem:subadditivity}), and Proposition \ref{prop:classical_many_to_one_additive} shows that the optimal cost under the classical many-to-one model is additive. To our best knowledge, these observations have not appeared in prior work. 
 

 
\begin{example}\label{example:Ulukus_eg1}
    Let  $\mathrm{LC}= (\bbF_2,5,\bfV_1,\dots,\bfV_5)$ where
    \begin{align*}
        \bfV_1 = \begin{bmatrix}
            0&1\\0&1\\0&0\\1&0
        \end{bmatrix},
        \bfV_2 = \begin{bmatrix}
            1&1\\1&0\\0&1\\1&1
        \end{bmatrix},
        \bfV_3 = \begin{bmatrix}
            0&0\\1&0\\1&1\\0&1
        \end{bmatrix},
        \bfV_4 = \begin{bmatrix}
            0 \\0 \\0\\1 
        \end{bmatrix},
        \bfV_5 = \begin{bmatrix}
            0 \\0 \\1\\0 
        \end{bmatrix}.
    \end{align*}
We note that the only $1\times 1$ invertible matrix in $\bbF_2$ is $1$. Therefore, the precoding problem in \cite{HuUlukus25} can be written as 
\begin{align*}
\textbf{Prob}_2(\text{LC})  =   \min_{(P_1,P_2,P_3)} \text{rank}\left( \sum_{i\in[3]}\bfV_i P_i \bfV_i^T + \begin{bmatrix}
        0&0&0&0\\0&0&0&0\\0&0&1&0\\0&0&0&1
    \end{bmatrix} \right),
\end{align*}
where $P_1,P_2,P_3$ each range over $2\times 2$ invertible matrices in $\bbF_2$. There are six such matrices. Therefore, there are $6^3=216$ cases. Exhaustive enumeration of the 216 cases yields that $\textbf{Prob}_2(\text{LC}) = 2$ (code available at \cite{Ruoyu25}). Therefore, the cost achieved by the scheme of \cite{HuUlukus25} is $\frac{2+8}{2} = 5$.

In contrast, our precoding strategy is strictly better. Set $s=1$, and
\begin{align*}
    P_2 = \begin{bmatrix}
        0&1\\1&1
    \end{bmatrix},
    P_3 = \begin{bmatrix}
        1&0\\0&1
    \end{bmatrix}, P_4=P_5=[1], \tilde{M}_1 = \begin{bmatrix}
        \overline{M} \\ 0_{4\times 3}
    \end{bmatrix}, \tilde{M}_2 = \begin{bmatrix}
     0_{4\times 3} \\ \overline{M}
    \end{bmatrix}, \tilde{X}_1 = \begin{bmatrix}
        \overline{X}&0_{3\times 2}\\0_{3\times 2}&\overline{X}
    \end{bmatrix},
\end{align*}
where 
\begin{align*}
    \overline{M} = \begin{bmatrix}
        1&1&1\\1&1&1\\ 1&1&0\\ 1&0&0
    \end{bmatrix}, \text{~and~} \overline{X} = \begin{bmatrix}
        1&0\\1&0\\0&1
    \end{bmatrix}.
\end{align*}
By calculation, we verify that (1) $\bfV_1 = \overline{M}\overline{X}$ and (2) $\Omega(\tilde{M}_1,\tilde{M}_2) + \sum_{i\in[k]\setminus \{1\}}  \Omega\Big(\begin{bmatrix}
            \bfV_iP_i\\ 0
        \end{bmatrix}, \begin{bmatrix}
            0 \\ \bfV_i
        \end{bmatrix} \Big) =0.$ Therefore, $((P_i)_{i=2}^5,\tilde{M}_1,\tilde{M}_2,\tilde{X}_1)$  is   feasible  for $s=1$. This implies that our cost is at most $\frac{1+8}{2}<5$, which is strictly smaller than the cost of \cite{HuUlukus25}.
\end{example}

\begin{example}\label{example:Ulukus_eg2}
Consider $\text{LC}=\left(\bbF_2, 4, \bfV_1=\begin{bmatrix}
    1\\0
\end{bmatrix}, \bfV_2=\begin{bmatrix}
    0\\1
\end{bmatrix}, \bfV_3=\begin{bmatrix}
    1\\1
\end{bmatrix}, \bfV_4=\begin{bmatrix}
    1\\0
 \end{bmatrix}\right)$. Here each Alice$_i$ has a scalar $W_i$, so   $m_i=1$. 
Therefore, each $P_i$ in $\text{Prob}_2(\text{LC})$ must be a nonzero scalar in $\bbF_2$, i.e.,  $P_i=1$ for all $i$. In this case, we have 
\begin{align*}
    \text{rank}\left(\sum_i\bfV_i P_i \bfV_i^T\right) =  \text{rank}\left(\begin{bmatrix}
        1&1\\1&0
    \end{bmatrix}\right)=2.
\end{align*}
Thus, the total cost of the scheme of \cite{HuUlukus25} is $\frac{2+4}{2}=3$.

On the other hand, note that this LC is a Double-Basis LC, because $[\bfV_1|\dots|\bfV_4]$ contains disjoint submatrices $\begin{bmatrix}
    1&0\\0&1
\end{bmatrix}$ and $\begin{bmatrix}
    1&1\\1&0
\end{bmatrix}$. Therefore, the coding scheme in Theorem \ref{theorem:first_main_result_informal}(b) applies, which gives a total cost of $\frac{2}{3}(\sum_i m_i ) = \frac{8}{3}< 3$. Thus our Restricted Interval-$[0,1]$ scheme can also strictly improve on the scheme of \cite{HuUlukus25}.
\end{example}

\section{Conclusions and Future Work}
\label{sec:conclusions}

In this work, we studied linear computation (LC) over a noiseless quantum many-to-one network. We showed that the $(\kappa,N)$-Sum Box construction, combined with suitable precoding, can be used as a general building block for coding schemes beyond the original Sum Box setting. In particular, we developed coding schemes for several classes of LC problems, including Interval-$[0,1]$ LC problems through an optimal-precoding formulation, and Restricted Interval-$[0,1]$ LC problems through a decomposition of the form $Y=U+\bfP V+\bfQ Z$. For both these classes of problems we provide schemes with guaranteed costs, that are strictly lower than prior known costs in certain cases. Furthermore, we discuss a class of direct-sum LC problems for which the quantum many-to-one cost can be strictly subadditive while the corresponding classical cost remains additive.

Several questions remain open. A central direction is to better understand the SSO condition $\Omega(M_x,M_z)=0$ and its interaction with precoding. A better understanding of conditions under which the SSO obstruction can be cancelled may lead to more efficient protocols and may also simplify the search for feasible precoding matrices. Another important direction is to develop converse bounds. Finally, it would be interesting to determine whether the sufficient conditions identified in this work can be weakened or made necessary.

\appendices

\section{Proofs for the Preliminaries} 

\subsection{Proof of Proposition \ref{prop:SSO_proposition}}\label{prop:SSO_proposition:app}
\begin{proof}[Proof of $(a)$]
    \begin{align*}
        &\Omega([M_{x_1}|M_{x_2}],[M_{z_1}|M_{z_2}])\\
        =&M_{x_1}M_{z_1}^T+M_{x_2}M_{z_2}^T- M_{z_1}M_{x_1}^T -M_{z_2}M_{x_2}^T \\
        =&M_{x_1}M_{z_1}^T- M_{z_1}M_{x_1}^T+M_{x_2}M_{z_2}^T-M_{z_2}M_{x_2}^T\\
        =&\Omega(M_{x_1},M_{z_1} ) + \Omega(M_{x_2},M_{z_2} ).
    \end{align*}
\end{proof}
\begin{proof}[Proof of $(b)$]
    $\Omega(M_x,M_z)=M_{x}M_{z}^T- M_{z}M_{x}^T=-(M_{z}M_{x}^T- M_{x}M_{z}^T)=-\Omega(M_z,M_x)$.
\end{proof}
\begin{proof}[Proof of $(c)$]The statement follows directly from distributivity of matrix multiplication. Indeed, for compatible matrices,
\begin{align*}
&    \Omega(M_{x_1}+M_{x_2},M_z)
=
\Omega(M_{x_1},M_z)+\Omega(M_{x_2},M_z),
\end{align*} 
with a similar argument for linearity in the second argument.
Moreover, for any scalar \(a\in\bbF_q\),
\[
\Omega(aM_x,M_z)=a\Omega(M_x,M_z),
\qquad
\Omega(M_x,aM_z)=a\Omega(M_x,M_z).
\]
Hence \(\Omega(\cdot,\cdot)\) is bilinear.
\end{proof}
\begin{proof}[Proof of $(d)$]
    \begin{align*}
        &\Omega(DM_x,DM_z)=DM_{x}M_{z}^TD^T- DM_{z}M_{x}^TD^T \\
        &= D(M_{x}M_{z}^T- M_{z}M_{x}^T)D^T= D\Omega(M_x,M_z)D^T.
    \end{align*}
\end{proof}

\begin{proof}[Proof of $(e)$] 
We have $\Omega(M_x,M_z)^T=M_zM_x^T-M_xM_z^T=-\Omega(M_x,M_z).$
    Moreover, it is easy to verify that $(\Omega(M_x,M_z))_{ii}=(M_xM_z^{T})_{ii}-(M_zM_x^{T})_{ii}
=\sum_{k=1}^n (M_x)_{ik}(M_z)_{ik}-\sum_{k=1}^n (M_z)_{ik}(M_x)_{ik}=0$, (where $n$ is the number of columns).  This argument also covers the case
\(\text{char}(\bbF_q)=2\), where the zero-diagonal condition is essential.
\end{proof}

 \subsection{Proof of Proposition~\ref{prop:trivial_precoding_bound}}
\label{app:prop:trivial_precoding_bound}

\begin{proof}
Fix invertible matrices \(P_i\in\bbF_q^{m_i\times m_i}\), and write $A_i:=\begin{bmatrix}\bfV_iP_i\\0\end{bmatrix},
\,
B_i:=\begin{bmatrix}0\\\bfV_i\end{bmatrix},
\, i\in[k].$  Let $M:=\sum_{i\in[k]}\Omega(A_i,B_i).$ The matrix \(M\) is alternating, and hence \(\text{rank}(M)\) is even ({\it cf.} Definition \ref{def:alternating_matrix}). Write $\text{rank}(M)=2r.$ By Lemma~\ref{lem:alternating_omega_factorization}, there exist
\(Q,R\in\bbF_q^{2m\times r}\) such that $M=\Omega(Q,R).$  We now construct a feasible point of Problem~\ref{problem:precoding} with
parameter \(s=r\). Set
\[
\widetilde M_1:=[A_1\mid -Q],
\qquad
\widetilde M_2:=[B_1\mid R],
\qquad
\widehat P_i:=P_i,\quad i=2,\dots,k,
\]
and define
\[
\widetilde X_1:=
\begin{bmatrix}
P_1^{-1} & 0\\
0 & 0\\
0 & I_{m_1}\\
0 & 0
\end{bmatrix},
\]
where the row block sizes are \(m_1,r,m_1,r\), and the column block
sizes are \(m_1,m_1\). First,
\[
[\widetilde M_1\mid \widetilde M_2]\widetilde X_1
=
[A_1\mid -Q\mid B_1\mid R]
\begin{bmatrix}
P_1^{-1} & 0\\
0 & 0\\
0 & I_{m_1}\\
0 & 0
\end{bmatrix}
=
[A_1P_1^{-1}\mid B_1]
=
\begin{bmatrix}
\bfV_1 & 0\\
0 & \bfV_1
\end{bmatrix}.
\]
Thus the reconstruction constraint in Problem~\ref{problem:precoding}
is satisfied. Second, using the block additivity and bilinearity of \(\Omega(\cdot,\cdot)\),
\[
\Omega(\widetilde M_1,\widetilde M_2)
=
\Omega(A_1,B_1)-\Omega(Q,R)
=
\Omega(A_1,B_1)-M.
\]
Therefore,
\[
\Omega(\widetilde M_1,\widetilde M_2)
+
\sum_{i=2}^k \Omega(A_i,B_i)
=
\Omega(A_1,B_1)-M+\sum_{i=2}^k \Omega(A_i,B_i)
=0.
\]
Hence the SSO constraint in Problem~\ref{problem:precoding} is also satisfied. Thus \(\mathcal F(r)\neq\emptyset\), and consequently $\mathbf{Prob}_1(\mathrm{LC})\le r
=
\frac{1}{2}\text{rank} (M(P_1,\dots,P_k)).$ Since \(M(P_1,\dots,P_k)\) is a \(2m\times 2m\) matrix, \(\text{rank} (M(P_1,\dots,P_k))\le 2m\), and hence $ \mathbf{Prob}_1(\mathrm{LC})\le m.$
\end{proof}

\section{Step-by-Step Calculations}
\subsection{Proof of the block-precoding rank identity}
\label{app:block_precoding_rank_identity}

We prove the identity used in the proof of
Theorem~\ref{theorem:first_main_result_informal}(a). Recall that $M_x=[M_{x,1}|\cdots|M_{x,k}],\,     M_z=[M_{z,1}|\cdots|M_{z,k}],$  where $M_{x,1}=\tilde M_1, 
    M_{z,1}=\tilde M_2,$ and, for $i\in\{2,\dots,k\}$, 
    $$M_{x,i}=
    \begin{bmatrix}
        \bfV_iP_i\\
        0
    \end{bmatrix},\,
    M_{z,i}=
    \begin{bmatrix}
        0\\
        \bfV_i
    \end{bmatrix}.$$  Let $ N=s+\sum_{i=1}^k m_i$, and $\Pi$ be the $2N\times 2N$ column-permutation matrix such that $[M_x|M_z]\Pi
    =
    [M_{x,1}|M_{z,1}|M_{x,2}|M_{z,2}|\cdots|M_{x,k}|M_{z,k}].$  Define
\[
    R_{\mathrm{blk}}
    =
    \operatorname{diag}
    \left(
        \tilde X_1,
        \begin{bmatrix}
            P_2^{-1} & 0\\
            0 & I_{m_2}
        \end{bmatrix},
        \dots,
        \begin{bmatrix}
            P_k^{-1} & 0\\
            0 & I_{m_k}
        \end{bmatrix}
    \right),
    \qquad
    R:=\Pi R_{\mathrm{blk}}.
\]
Then $R$ has dimension
\[
    2\left(s+\sum_{i=1}^k m_i\right)
    \times
    2\sum_{i=1}^k m_i.
\]
By the defining constraint of the Optimal Precoding Problem,
\[
    [\tilde M_1|\tilde M_2]\tilde X_1
    =
    \begin{bmatrix}
        \bfV_1 & 0\\
        0 & \bfV_1
    \end{bmatrix}
    =\bfV_1^{[2]}.
\]
Moreover, for each $i\in\{2,\dots,k\}$,
\begin{align*}
    [M_{x,i}|M_{z,i}]
    \begin{bmatrix}
        P_i^{-1} & 0\\
        0 & I_{m_i}
    \end{bmatrix}
    =
    \begin{bmatrix}
        \bfV_iP_i & 0\\
        0 & \bfV_i
    \end{bmatrix}
    \begin{bmatrix}
        P_i^{-1} & 0\\
        0 & I_{m_i}
    \end{bmatrix}  =
    \begin{bmatrix}
        \bfV_i & 0\\
        0 & \bfV_i
    \end{bmatrix}
    =
    \bfV_i^{[2]}.
\end{align*}
Therefore,
\begin{align*}
    [M_x|M_z]R
    &=
    [M_x|M_z]\Pi R_{\mathrm{blk}}\\
    &=
    [M_{x,1}|M_{z,1}|M_{x,2}|M_{z,2}|\cdots|M_{x,k}|M_{z,k}]
    R_{\mathrm{blk}}\\
    &=
    [\bfV_1^{[2]}|\bfV_2^{[2]}|\cdots|\bfV_k^{[2]}].
\end{align*}
This proves the claimed block-precoding identity.

\subsection{Proof of  $ \Omega(M_x,M_z)=0$ for $M_x,M_z$ defined in \eqref{eq:Mx_Mz_of_eg_of_restricted}}\label{eq:restricted_interval_eq:app}
\begin{proof}  
We make the following observation in \eqref{eq:lemma7_eg_eq1}.
\begin{equation}
    \begin{split}\label{eq:lemma7_eg_eq1}
       & \Omega\left(\begin{bmatrix}
        \bfP&\bfQ\\0_{m\times m}& 0_{m\times \bar m}\\0_{m\times m}& 0_{m\times \bar m}
    \end{bmatrix},
    \begin{bmatrix}
        0_{m\times m}& 0_{m\times \bar m}\\\bfP&\bfQ\\0_{m\times m}& 0_{m\times \bar m}
    \end{bmatrix}\right)\\
   & \overset{(a)}{=} \Omega\left(\begin{bmatrix}
        \bfP \\0_{m\times m} \\0_{m\times m} 
    \end{bmatrix},
    \begin{bmatrix}
        0_{m\times m} \\\bfP \\0_{m\times m} 
    \end{bmatrix}\right) + \Omega\left(\begin{bmatrix}
         \bfQ\\  0_{m\times \bar m}\\  0_{m\times \bar m}
    \end{bmatrix},
    \begin{bmatrix}
       0_{m\times \bar m}\\ \bfQ\\  0_{m\times \bar m}
    \end{bmatrix}\right)\\
   & \overset{(b)}{=} \begin{bmatrix}
        0_{m\times m} & \bfP\bfP^T & 0_{m\times m}\\
        - \bfP\bfP^T & 0_{m\times m}& 0_{m\times m}\\
         0_{m\times m}& 0_{m\times m}& 0_{m\times m}
    \end{bmatrix}+\begin{bmatrix}
        0_{m\times m} & \bfQ\bfQ^T & 0_{m\times m}\\
        - \bfQ\bfQ^T & 0_{m\times m}& 0_{m\times m}\\
         0_{m\times m}& 0_{m\times m}& 0_{m\times m}
    \end{bmatrix}\\
   & = \begin{bmatrix}
        0_{m\times m} & \bfP\bfP^T +\bfQ\bfQ^T& 0_{m\times m}\\
        - \bfP\bfP^T -\bfQ\bfQ^T& 0_{m\times m}& 0_{m\times m}\\
         0_{m\times m}& 0_{m\times m}& 0_{m\times m}
    \end{bmatrix}
    \overset{(c)}{=}\Omega\left(\begin{bmatrix}
        I_m\\ 0_{m\times m}\\ 0_{m\times m}
    \end{bmatrix},
    \begin{bmatrix}
     0_{m\times m}\\\bfP\bfP^T +\bfQ\bfQ^T\\ 0_{m\times m}    
    \end{bmatrix}\right),
    \end{split}
\end{equation}
where $(a)$ is by Proposition \ref{prop:SSO_proposition}\ref{prop:SSO_partition}, and $(b)$ and $(c)$ are by  Definition \ref{def:SSO_matrix} and calculation.
\begin{equation}
    \begin{split}\label{eq:lemma7_eg_eq2}
        & \Omega\left(\begin{bmatrix}
        0_{m\times m}& 0_{m\times \bar m}\\ I_m &  0_{m\times \bar m}\\0_{m\times m}& 0_{m\times \bar m}
    \end{bmatrix},
    \begin{bmatrix}
        0_{m\times m}& 0_{m\times \bar m}\\0_{m\times m}& 0_{m\times \bar m}\\\bfP&\bfQ
    \end{bmatrix}\right)\\
    \overset{(a)}{=}&  \Omega\left(\begin{bmatrix}
        0_{m\times m} \\ I_m  \\0_{m\times m} 
    \end{bmatrix},
    \begin{bmatrix}
        0_{m\times m} \\0_{m\times m} \\\bfP 
    \end{bmatrix}\right) 
    + \Omega\left(\begin{bmatrix}
         0_{m\times \bar m}\\ 0_{m\times \bar m}\\  0_{m\times \bar m}
    \end{bmatrix},
    \begin{bmatrix}
        0_{m\times \bar m}\\   0_{m\times \bar m}\\ \bfQ
    \end{bmatrix}\right)\\
    \overset{(b)}{=}& \begin{bmatrix}
         0_{m\times m} & 0_{m\times m} & 0_{m\times m} \\
          0_{m\times m} & 0_{m\times m} & \bfP^T\\
           0_{m\times m} &-\bfP &  0_{m\times m} 
    \end{bmatrix} + 0 \overset{(c)}{=}  \Omega\left(\begin{bmatrix}
         0_{m\times  m}\\ 0_{m\times  m}\\  I_m
    \end{bmatrix},
    \begin{bmatrix}
        0_{m\times  m}\\   -\bfP^T\\ 0_{m\times  m}
    \end{bmatrix}\right)
    \end{split}
\end{equation}

\begin{align*}
    0  &\overset{(a)}{=}   \Omega\left(\begin{bmatrix}
        \bfP&\bfQ\\0_{m\times m}& 0_{m\times \bar m}\\0_{m\times m}& 0_{m\times \bar m}
    \end{bmatrix},
    \begin{bmatrix}
        0_{m\times m}& 0_{m\times \bar m}\\\bfP&\bfQ\\0_{m\times m}& 0_{m\times \bar m}
    \end{bmatrix}\right) - \Omega\left(\begin{bmatrix}
        I_m\\ 0_{m\times m}\\ 0_{m\times m}
    \end{bmatrix},
    \begin{bmatrix}
     0_{m\times m}\\\bfP\bfP^T +\bfQ\bfQ^T\\ 0_{m\times m}    
    \end{bmatrix}\right)\\
    &+ \Omega\left(\begin{bmatrix}
        0_{m\times m}& 0_{m\times \bar m}\\ I_m &  0_{m\times \bar m}\\0_{m\times m}& 0_{m\times \bar m}
    \end{bmatrix},
    \begin{bmatrix}
        0_{m\times m}& 0_{m\times \bar m}\\0_{m\times m}& 0_{m\times \bar m}\\\bfP&\bfQ
    \end{bmatrix}\right) - 
    \Omega\left(\begin{bmatrix}
         0_{m\times  m}\\ 0_{m\times  m}\\  I_m
    \end{bmatrix},
    \begin{bmatrix}
        0_{m\times  m}\\   -\bfP^T\\ 0_{m\times  m}
    \end{bmatrix}\right)\\ 
    &\overset{(b)}{=}   \Omega\left(\begin{bmatrix}
        \bfP&\bfQ\\0_{m\times m}& 0_{m\times \bar m}\\0_{m\times m}& 0_{m\times \bar m}
    \end{bmatrix},
    \begin{bmatrix}
        0_{m\times m}& 0_{m\times \bar m}\\\bfP&\bfQ\\0_{m\times m}& 0_{m\times \bar m}
    \end{bmatrix}\right) + \Omega\left(
    \begin{bmatrix}
     0_{m\times m}\\\bfP\bfP^T +\bfQ\bfQ^T\\ 0_{m\times m}    
    \end{bmatrix},\begin{bmatrix}
        I_m\\ 0_{m\times m}\\ 0_{m\times m}
    \end{bmatrix}\right)\\
    &+ \Omega\left(\begin{bmatrix}
        0_{m\times m}& 0_{m\times \bar m}\\ I_m &  0_{m\times \bar m}\\0_{m\times m}& 0_{m\times \bar m}
    \end{bmatrix},
    \begin{bmatrix}
        0_{m\times m}& 0_{m\times \bar m}\\0_{m\times m}& 0_{m\times \bar m}\\\bfP&\bfQ
    \end{bmatrix}\right) + 
    \Omega\left(\begin{bmatrix}
         0_{m\times  m}\\ 0_{m\times  m}\\  I_m
    \end{bmatrix},
    \begin{bmatrix}
        0_{m\times  m}\\    \bfP^T\\ 0_{m\times  m}
    \end{bmatrix}\right)\\
    &\overset{(c)}{=} \Omega\left(\begin{bmatrix}
        \bfP & \bfQ & 0_{m\times m} & 0_{m\times m} & 0_{m\times \bar m}  & 0_{m\times  m}\\
        0_{m\times m} & 0_{m\times \bar m} & \bfP\bfP^T +\bfQ\bfQ^T & I_m & 0_{m\times \bar m} & 0_{m\times  m}\\
        0_{m\times m} & 0_{m\times \bar m} & 0_{m\times m} & 0_{m\times m} & 0_{m\times \bar m} & I_m
    \end{bmatrix}, \right.
    \notag\\&\qquad
    \left. 
    \begin{bmatrix}
        0_{m\times m} & 0_{m\times \bar m}  & I_m & 0_{m\times m} & 0_{m\times \bar m}  & 0_{m\times   m}\\
        \bfP & \bfQ& 0_{m\times m} &0_{m\times m} & 0_{m\times \bar m} & \bfP^T\\
        0_{m\times m} & 0_{m\times \bar m} & 0_{m\times m} &  \bfP & \bfQ & 0_{m\times   m}
    \end{bmatrix}
    \right)\\
    &= \Omega(M_x,M_z),
\end{align*}
where  $(a)$ follows from \eqref{eq:lemma7_eg_eq1} and \eqref{eq:lemma7_eg_eq2}, $(b)$ follows from Proposition \ref{prop:SSO_proposition}\ref{prop:SSO_bilinear} and \ref{prop:SSO_proposition}\ref{prop:SSO_asymmetric}, $(c)$ follows from Proposition \ref{prop:SSO_proposition}\ref{prop:SSO_partition}. 
\end{proof}

\section{Proof of Proposition \ref{prop:Restricted Interval-$[0,1]$_two_disjoint_submatrices_condition}}  \label{app:prop:Restricted Interval-$[0,1]$_two_disjoint_submatrices_condition}

\subsection{Proof of Part (a)} 
\begin{proof}
It suffices to find $L$ and choose $A_i$ from each $\bfV_i^{[L]}$. Set $L=2$  and let $D$ and $D'$ denote two column-disjoint $m\times m$
full-rank submatrices of $\bfV$. For each $i\in[k]$, suppose that $D_i,D_i'$ are the sub-matrices of $D,D'$ that  are contained in $\bfV_i$, i.e., they are sub-matrices of $\bfV_i$. Therefore, $D=[D_1|\dots|D_k]$ and $D'=[D_1'|\dots|D_k']$, and for each $i\in [k]$, $\begin{bmatrix}
        D_i&0\\0&D_i'
    \end{bmatrix}$ is a submatrix of $\bfV_i^{[2]} =\begin{bmatrix}
    \bfV_i&0\\
    0& \bfV_i
\end{bmatrix}$. We set $A_i = \begin{bmatrix}
        D_i&0\\0&D_i'
    \end{bmatrix}$.  Since $D,D'$ are column-disjoint submatrices of $\bfV$, it follows that $D_i,D_i'$ are column-disjoint submatrices of $\bfV_i$.  Then, $A_i$  has at most $m_i$ columns, i.e. $a_i\le \frac{2 m_i}{2}$.  

Moreover, since $A_i=\begin{bmatrix}
        D_i&0\\0&D_i'
    \end{bmatrix}$ is block-diagonal, then
    \begin{align*}
        &\text{rank}([A_1|A_2|\dots|A_k]) \\
        &= \text{rank} \left(\begin{bmatrix}
            D_1 & 0 &D_2 &0 &\dots & D_k &0\\
             0 & D_1' &0 & D_2' &\dots & 0 & D_k'\\
        \end{bmatrix} \right)\\
        &=\text{rank}([D_1|\dots|D_k]) +\text{rank}([D_1'|\dots|D_k'])
        =\text{rank}(D) +\text{rank}(D')
        =2m.
    \end{align*} 
\end{proof}

\subsection{Proof of Part (b)}   
In this part we use the formalism of matroids \cite{schrijver2003combinatorial}.
A matroid \(M=(\mathcal N,\mathcal I)\) consists of a finite ground set \(\mathcal N\) and a family \(\mathcal I\) of independent sets. A basis of \(M\) is a maximal independent set; all bases of a matroid have the same cardinality, called the rank of the matroid. Given a matrix \(\bfV\), we denote by \(M(\bfV)\) the vector matroid represented by the columns of \(\bfV\): the ground set is the set of column indices, and a set of indices is independent if and only if the corresponding columns are linearly independent. Thus, if \(\text{rank}(\bfV)=m\), the bases of \(M(\bfV)\) are exactly the full-rank \(m\)-subsets of columns. An independence oracle for a matroid is an oracle that answers whether a queried set belongs to \(\mathcal I\). For vector matroids, this oracle is implemented by Gaussian elimination.

The \(p\)-fold matroid union problem asks for \(p\) bases \(B_1,\dots,B_p\) maximizing the size of their union; with capacities, each element \(e\) is counted only up to its capacity \(u(e)\). In our application, we use unit capacities \(u(e)=1\), which means that each column can be counted at most once. Thus, achieving objective value \(2m\) for two bases is equivalent to finding two non-overlapping full-rank \(m\)-subsets of columns, i.e.,
\[
\text{Double-Basis}
\quad\Longleftrightarrow\quad
\exists\ B_1,B_2 \text{ bases of } M(\bfV)
\text{ with } B_1\cap B_2=\varnothing .
\]
Equivalently, this holds if and only if the 2-fold matroid union objective with unit capacities \(u(e)=1\) has value \(2m\). Intuitively, \(u(e)=1\) is exactly the non-overlap constraint: a column cannot be used by both bases. 
\begin{proof} 
Let $N:=\sum_{i\in[k]}m_i$ and $\bfV=[\bfV_1|\bfV_2|\cdots|\bfV_k]\in \bbF_q^{m\times N}$. By Assumption \ref{assumption:lin_ind_assumption},   $\text{rank}(\bfV)=m$ holds.  Consider the vector matroid $M(\bfV)=([N],\,\mathcal I),$  where
\[
    S\in\mathcal I
    \quad\Longleftrightarrow\quad
    \text{rank}(\bfV_S)=|S|,
\]
where $\bfV_S$ consists of columns of $\bfV$ from the index set $S$. The rank of this matroid is \(r=m\), and its bases are exactly the full-rank \(m\)-subsets of columns of \(\bfV\). We use the standard exact algorithm for \(p\)-fold matroid union, where \(p\) denotes the number of copies of the matroid. 

\begin{problem}[$p$-fold Matroid Union {\cite[Section~3]{quanrud2024faster}}]
    Suppose we have a matroid \(M=(\mathcal N,\mathcal I)\), integer capacities \(u:\mathcal N\to\mathbb N\), and an integer \(p\). For bases \(B_1,\dots,B_p\), let $x(e):=|\{j\in[p]:e\in B_j\}|.$  The \(p\)-fold matroid union is the following maximization problem 
\[
    \mathrm{OPT}:=\max_{B_1,\dots,B_p}
    \sum_{e\in\mathcal N}\min\{u(e),x(e)\},
\]
where $B_1,\dots,B_p$ range over bases of the matroid.
\end{problem}

\begin{theorem}
    Quanrud \cite[Theorem~3.6]{quanrud2024faster} shows that, for integer capacities, a maximum \(p\)-fold Matroid Union can be computed using $O(n+\mathrm{OPT}\cdot r\log(pr))$ independence-oracle queries where $r$ is the rank of the matroid. 
\end{theorem} 

We apply this result to the vector matroid \(M(\bfV)\) with $p=2,\, u(e)=1\quad\text{for all }e\in \calN.$ 
For two bases \(B_1,B_2\) of \(M(\bfV)\), the objective becomes $\sum_{e\in \calN}\min\{1,x(e)\} = |B_1\cup B_2|.$  Since each base has size \(m\), we have $|B_1\cup B_2|\le |B_1|+|B_2|=2m.$  Moreover, $|B_1\cup B_2|=2m \Longleftrightarrow  B_1\cap B_2=\varnothing.$  Therefore, the original instance is Double-Basis if and only if the optimum value of the above 2-fold matroid union instance is \(2m\).

It remains to bound the running time. In our application, $n=N=\sum_{i\in[k]}m_i,\,r=m,\,p=2,\,\mathrm{OPT}\le 2m.$ Thus Theorem~3.6 of \cite{quanrud2024faster} gives $O(N+2m^2\log(2m))$ independence-oracle queries. For the vector matroid \(M(\bfV)\), an independence-oracle query asks whether a set of columns \(S\subseteq[N]\) is linearly independent, i.e., $\text{rank}(\bfV_S)=|S|.$ This can be checked by Gaussian elimination over \(\bbF_q\) in polynomial time in \(m\) and \(N\). Hence the Double-Basis property can be decided in $O\!\left(\text{poly}\!\left(m,\sum_{i\in[k]}m_i\right)\right)$ time. 
\end{proof}

\subsection{Proof of Part (c)} 
\begin{proof} 
Let $m,m_1,\dots,m_k$ be fixed such that $\frac{2m}{\sum_{i\in[k]} m_i}  \in[0,1]$. 
We restate the relevant conditions here.
\begin{enumerate}[label=(\arabic*)]
    \item For each $i \in [k]$, the columns of $\bfV_i$ are all linearly independent.
    \item $\text{rank}([\bfV]) = m$.
    \item The instance is Double-Basis: for
    $\bfV=[\bfV_1|\cdots|\bfV_k]$, the matrix $\bfV$ contains two
    disjoint $m\times m$ full-rank submatrices $D,D'$.
\end{enumerate}
Recall that $\text{Num}_A$ is the number of LC problems satisfying  (1), (2) and (3),  
$\text{Num}_B$ is the number of LC problems satisfying  (1), (2). Let $\rho =  {\text{Num}_A}/{\text{Num}_B}$, and  $\text{Num}_C$ be the total number of LC problems $(\bbF_q,k,\bfV_1,\dots,\bfV_k)$ such that it  does not have to satisfy any of (1), (2) and (3), but $\bfV_i$ must be of dimension $m\times m_i$.

We sample each LC problem contributing to $\text{Num}_C$ by generating each column of  $\bfV$ i.i.d. with distribution $\text{Uniform}(\bbF_q^m)$. Let $\text{Event}_B$ denote the event that (1) and (2) are satisfied, and $\text{Event}_A$ the event that (1), (2) and (3) are satisfied. Then, we have $\text{Pr}(\text{Event}_A) = \text{Num}_A / \text{Num}_C,$ and $\text{Pr}(\text{Event}_B) = \text{Num}_B / \text{Num}_C$. It follows that 
\begin{align*}
    \rho =  {\text{Num}_A}/{\text{Num}_B}= \frac{\text{Pr}(\text{Event}_A)}{\text{Pr}(\text{Event}_B)}\ge \text{Pr}(\text{Event}_A).
\end{align*}
The argument below establishes that $\mathrm{Pr}(\mathrm{Event}_A)\ge 1- \frac{k+2}{q-1}$.

$\text{Event}_A$ is the event that (1), (2) and (3) all hold. We note that (3) implies (2), because $\text{rank}(\bfV)\le m$ as $\bfV$ has $m$ rows. Define $\text{Event}_C$ to be the event that (1) happens and (2), (3) may or may not happen, and define $\text{Event}_{C'}$ to be the event that the first $m$  and the second $m$ columns of $\bfV$ are valid choices of $D,D'$. Then, $\text{Event}_{C'}$ implies (3), but (3) may not imply $\text{Event}_{C'}$. Therefore, 
\begin{equation}\label{eq:event_c_c'_implies_a}
    \text{If both $\text{Event}_C$ and $\text{Event}_{C'}$ occur, then $\text{Event}_{A}$ occurs.}
\end{equation}  

Recall $\#\text{FullCol}(m,n;q)$ denotes the number of $m\times n$ matrices over $\bbF_q$ with full column rank  (\textit{cf.} Proposition \ref{prop:number_of_FullCol}). Since $\frac{2m}{\sum_{i\in[k]}m_i}\in[0,1]$, there are at least $2m$ columns generated.  At the moment $2m$ columns are generated, the probability that $\text{Event}_{C'}$ happens is lower bounded by 
\begin{align*}
    &\text{Pr($\text{Event}_{C'}$)}\\
    \overset{(a)}{=}&\frac{\#\text{FullCol}(m,m;q)}{|\bbF_q^m|^m} \cdot \frac{\#\text{FullCol}(m,m;q)}{|\bbF_q^m|^m}\\
    \overset{(b)}{=}&\frac{\prod_{i=0}^{m-1}(q^m-q^i)}{q^{m\cdot m}} \cdot \frac{\prod_{i=0}^{m-1}(q^m-q^i)}{q^{m\cdot m}}\\
    =&\Big ( \prod_{i=1}^{m}(1-q^{-i}) \Big )^2\\
    \ge&\Big ( \prod_{i=1}^{\infty}(1-q^{-i}) \Big )^2,
\end{align*} 
where $(a)$ is because there are $|\bbF_q^m|^m$ ways to choose $m$ columns and $\#\text{FullCol}(m,m;q)$ ways to generate $m$ linearly independent columns, and $(b)$ is because of Proposition \ref{prop:number_of_FullCol}.

The probability that $\text{Event}_{C}$ occurs is lower bounded by
\begin{align*}
    &\text{Pr($\text{Event}_{C}$)} \\
    \overset{(a)}{=}& \frac{\#\text{FullCol}(m,m_1;q)}{|\bbF_q^m|^{m_1}}\cdot\dots\cdot \frac{\#\text{FullCol}(m,m_k;q)}{|\bbF_q^m|^{m_k}}\\
    \overset{(b)}{=}& \prod_{i=1}^k\frac{\prod_{j=0}^{m_i-1}(q^m-q^j)}{q^{m\cdot m_i}}\\
    =&  \prod_{i=1}^k \Big(\prod_{j=1}^{m_i}(1-q^{-j}) \Big)\\
    \ge&  \Big ( \prod_{i=1}^{\infty}(1-q^{-i}) \Big )^k,
\end{align*}
where $(a)$ is because there are $|\bbF_q^m|^{m_i}$ ways to choose $m_i$ columns and $\#\text{FullCol}(m,m_i;q)$ ways to generate $m_i$ linearly independent columns, and $(b)$ is because of Proposition \ref{prop:number_of_FullCol}.  Note 
\begin{align*}
  1-\frac{1}{q-1} = 1 -  \sum_{i=1}^{\infty} q^{-i} \le    \prod_{i=1}^{\infty}(1-q^{-i}) \le 1
\end{align*} where $1 -  \sum_{i=1}^{\infty} q^{-i} \le    \prod_{i=1}^{\infty}(1-q^{-i})$ is by the fact that $\forall N\ge 1,\,    1 -  \sum_{i=1}^{N} q^{-i} \le    \prod_{i=1}^{N}(1-q^{-i}),$ and taking $N\to \infty$. In particular, this shows that $|\prod_{i=1}^{\infty}(1-q^{-i})-1|\le \frac{1}{q-1}.$ Then, for $q\ge 2$, we have
\begin{align*}
&\text{Pr($\text{Event}_A$)}\\
    &\overset{(a)}{\ge} \text{Pr($\text{Event}_C$ and $\text{Event}_{C'}$)}\\ 
    &\ge \text{Pr($\text{Event}_C$)} + \text{Pr($\text{Event}_{C'}$)} - 1\\
    &\overset{(b)}{\ge}  \Big ( \prod_{i=1}^{\infty}(1-q^{-i}) \Big )^k+\Big ( \prod_{i=1}^{\infty}(1-q^{-i}) \Big )^2-1\\
    &\overset{(c)}{\ge}   \left(1 - \frac{1}{q-1}\right)^k + \left(1-\frac{1}{q-1}\right)^2-1\\
    &\overset{(d)}{\ge}  2-  \frac{2}{q-1}-  \frac{k}{q-1} -1 = 1 - \frac{k+2}{q-1},
\end{align*}
where  $(a)$ is by \eqref{eq:event_c_c'_implies_a}, the second inequality is the lower bound on the probability of the intersection of two events, $(b)$ follows from the lower bounds on $\text{Pr($\text{Event}_C$ happens)}$ and $\text{Pr($\text{Event}_{C'}$ happens)}$, $(c)$ is by assumption that $|\prod_{i=1}^{\infty}(1-q^{-i}) - 1|\le \frac{1}{q-1}$, which implies that $\prod_{i=1}^{\infty}(1-q^{-i})\ge 1-\frac{1}{q-1}$, and $(d)$ is by $(1-x)^t \ge 1-tx$ for $x>0$.    

\end{proof}

\bibliographystyle{IEEEtran}

\end{document}